\documentclass[12pt]{amsart}
\usepackage{amsmath,amssymb, amscd, graphicx}
\usepackage{graphicx}
\usepackage[all]{xy}
\usepackage{color}

\makeatletter
\@addtoreset{equation}{section}
\makeatother
\usepackage{enumerate}

\newcommand{\revsMM}[1]{{#1}}
\newcommand{\revsS}[1]{{#1}}
\newcommand{\revsSS}[1]{{#1}}
\newcommand{\revsSSS}[1]{{#1}}
\newcommand{\revsSSSS}[1]{{#1}}
\newcommand{\revsMMM}[1]{{#1}}

\newcommand{\ins}[1]{{#1}}
\newcommand{\revs}[1]{{#1}}
\newcommand{\insS}[1]{{#1}}
\newcommand{\insSS}[1]{{#1}}
\newcommand{\insSSS}[1]{{#1}}

\newcommand{\insM}[1]{{#1}}
\newcommand{\insOS}[1]{{#1}}

\newcommand{\insOSS}[1]{{#1}}

\newcommand{\insMMM}[1]{{#1}}
\newcommand{\insOSSS}[1]{{#1}}

\def\Im{\mathop{\mathrm{Im}}}
\def\Re{\mathop{\mathrm{Re}}}

\def\ii{{\sqrt{-1}}}

\def\cS{{\mathcal{S}}}
\def\cE{{\mathcal{E}}}

\def\fLSCp{{\mathcal{A}}_p}
\def\fLSCd{{\mathcal{A}}_d}
\def\fLBCCd{{\mathcal{B}}}
\def\AABCCd{{\mathbb{B}}}
\def\AASCd{{\mathbb{A}_d}}
\def\RSCT{{\mathcal{R}_d}}
\def\RBCCT{{\mathcal{R}_3}}
\def\RBCCd{{\mathcal{R}}}
\def\fD{{\mathcal{D}}}
\def\cM{{\mathcal{M}}}
\def\cF{{\mathcal{F}}}

\def\cF{{\mathcal{F}}}

\def\Ker{\mathrm {Ker}}
\def\diag{\mathrm {diag}}

\def\AA{{\mathbb A}}

\def\CC{{\mathbb C}}
\def\ZZ{{\mathbb Z}}

\def\RR{{\mathbb R}}
\def\EE{{\mathbb E}}

\def\cF{{\mathcal{F}}}

\def\UU{{\mathrm{U}}}
\def\Path{{\mathrm{Path}}}

\def\SO{{\mathrm{SO}}}

\def\SC{{\mathrm{SC}}}
\def\BCC{{\mathrm{BCC}}}
\def\abs#1{\left|#1\right|}
\def\diag{\mathrm {diag}}

\def\LA{\langle}
\def\RA{\rangle}

\theoremstyle{plain}
\newtheorem{theorem}{Theorem}[section]
\newtheorem{proposition}[theorem]{Proposition}
\newtheorem{lemma}[theorem]{Lemma}

\theoremstyle{definition}
\newtheorem{definition}[theorem]{Definition}

\theoremstyle{remark}
\newtheorem{remark}[theorem]{Remark}

\def\dfrac#1#2{{\displaystyle\frac{#1}{#2}}}

\def\book#1{\rm{#1},}
\def\paper#1{\textit{#1},}
\def\jour#1{\rm{#1}}
\def\yr#1{({\rm{#1})}}
\def\vol#1{\textbf{#1}}
\def\pages#1{\rm{#1}}

\def\publ#1{\rm{#1},}
\def\by#1{{\rm{#1},}}

\begin{document}

\title{An algebraic description of screw dislocations 
in SC and BCC crystal lattices}

\author{Hiroyasu Hamada \and
Shigeki Matsutani \and
Junichi Nakagawa \and
Osamu Saeki
 \and
Masaaki Uesaka
}




\maketitle

\begin{abstract}
We give an algebraic description of screw dislocations in
a crystal, especially simple cubic (SC) and body centered cubic (BCC) crystals, using
free abelian groups and fibering structures.
We also show that the \ins{strain}
energy of a screw dislocation based on the spring model
is expressed by the Epstein\ins{-Hurwitz} zeta function approximately.
\end{abstract}

\section{Introduction}

Mathematical descriptions of dislocations in crystal lattices have
been studied extensively in the framework of differential geometry or
continuum geometry
\insOSS{
\cite{A1,A2,KE,Kon,N}.
}
However, crystal structures are usually given as discrete geometry. In fact,
a crystal lattice is governed by a discrete group such as
free abelian group \cite{CS,K,S}.
An abelian group and its group ring provide fruitful mathematical tools,
e.g., abelian varieties, theta functions,
and so on.
On the other hand, the continuum geometric nature of dislocations
in the euclidean space cannot be represented well as long as we use
the ordinary algebraic expressions.

In this article, we give algebraic descriptions of
screw dislocations in the
simple cubic (SC) and the body centered cubic (BCC) crystals
in terms of certain ``fibrations'' involving
group rings.
Using such fibering structures, we describe the
continuum geometric property embedded in the euclidean space,
whereas the algebraic structures involving group rings enable us
to describe the discrete nature of the crystal lattices. More precisely,
we first describe the screw dislocations in continuum picture
and then we use algebraic structures of lattices to embed their
``discrete dislocations'' in the continuum description.
Our key idea is to use certain sections of $S^1$-bundles
to control the behavior of dislocations.

We also show that the  \ins{strain}
energy of a screw dislocation based on the spring model
is expressed by the Epstein\ins{-Hurwitz} zeta function
\revs{\cite{E,El}} approximately.
This will be shown by using our algebraic description of a screw
dislocation.

It should be noted that the results presented in this article arose in
our attempt to
solve the following problems proposed in the
Study Group Workshop at Kyushu University and the University of Tokyo,
held during July 29--Aug 4, 2015 \cite{O}.
\begin{enumerate}
\item To find a proper mathematical description of a
screw dislocation in the BCC lattice.
\item To find a proper mathematical description of
the \ins{strain} energy of a screw dislocation in the BCC lattice.
\end{enumerate}

The present article is organized as follows.
In Section~\ref{section2}, we first introduce certain fibering structures over the plane
involving the celebrated exact sequence
$$0 \to \ZZ \to \RR \to \UU(1) \to 1,$$
which will be essential in this article. Then,
we will describe screw dislocations
in continuum picture, in which certain covering spaces of a
punctured complex plane will play important roles.
By using certain path spaces, we clarify the covering structures
of screw dislocations in Remark~\ref{rmk:2.3}
and Proposition~\ref{prop:2.6}.
In Section~\ref{section3},
we first explain the algebraic structure
of the SC lattice as a free abelian group. Then, we consider
the fibering structure of the SC lattice in terms of
the associated group ring and its quotient.
Using these descriptions, we will describe
the screw dislocations in the SC lattice
in Propositions~\ref{prop:SC-single} and \ref{prop:SC-multi}
so that they are embedded in the continuum picture
given in Section~\ref{section2}.
In Section~\ref{section4},
we first explain the fibering structure of
the SC lattice in a diagonal direction, using group
ring structures. Then, we explain similar structures
for the BCC lattice. Finally, screw dislocations in the BCC lattice
will be described using all these algebraic materials
in Propositions~\ref{prop:singleBCC}
and \ref{prop:multiBCC}.
Section~\ref{section5} is devoted to the computation
of the \ins{strain} or elastic
energy of a single screw dislocation in the SC lattice.
\insS{In Theorem~\ref{thm:energy},
the elastic energy of the dislocation corresponding
to a bounded annular region
is expressed in terms of
the truncated Epstein\ins{-Hurwitz} zeta function \cite{E,El} approximately.}
\insS{In Section~\ref{section6},
\insOSSS{
we give some remarks on our
results from mathematical viewpoints.
In Section~\ref{sec:PhysV},
we explain and interpret our results in details
from physical points of view so that even the readers
who may not be familiar with mathematical tools can understand the novelty of
our results.}
Finally, in the \insSSS{appendix},
we show that the total elastic energy
for \insOSSS{a pair of} parallel screw dislocations with opposite
directions
does not diverge for unbounded regions \revsSSS{in the
continuum picture, and we also show the convergence
in the discrete picture for the
principal part of the energy}.}


\subsection{Notations and Conventions}

Throughout the article, we distinguish the euclidean space $\EE$ from
the real vector space $\RR$:
in particular, $\RR$ is endowed with an algebraic structure,
while $\EE$ is not. We often identify the $2$-dimensional
euclidean space $\EE^2$ with the complex plane $\CC$.
The group $U(1)$ acts on the set $S^1$ simply transitively.
Given a fiber bundle $\cF \rightarrow \cM$ over a base space
$\cM$,
we denote the set of all continuous sections $f:\cM \to \cF$
by $\Gamma(\cM, \cF)$.

\section{Screw Dislocations in Continuum Picture}\label{section2}

Let us consider the exact sequence of groups (see \cite{B})
\begin{equation}
\xymatrix@!C=50pt{
0 \ar[r] & \ZZ\ \ar[r]^-\iota & \RR
        \ar[r]^-{\exp 2\pi \ii} & \UU(1) \ar[r]^-{} & 1,}
\label{eq1}
\end{equation}
where $\ZZ$ and $\RR$ are additive groups, $\UU(1)$ is a
multiplicative group,
$\iota(n)=n$ for $n \in \ZZ$, and
$(\exp 2\pi \ii)(x) = \exp (2\pi \ii x)$ for $x \in \RR$.
Throughout this section, we fix $d > 0$.
For $\delta \in \RR$, we have
the ``shifts''
\begin{eqnarray*}
& & \widetilde\iota_{\delta}:\RR \to \EE \mbox{ defined by }
x \mapsto d\cdot x +\delta, \, x \in \RR, \mbox{ and} \\
& & \iota_{\delta}:\UU(1) \to S^1 \mbox{ defined by }
\exp({\ii\theta}) \mapsto
\exp{\ii(\theta +2\pi\delta/d)}, \, \theta \in \RR,
\end{eqnarray*}
such that the following diagram
is commutative:
$$
\xymatrix@!C=50pt{
 & & \EE\ \ar[r]^{\psi} & S^1\\
0 \ar[r] & \ZZ  \ar[r]^-{\iota} \ar[ru]^-{\varphi_{\delta}}
& \RR \ar[u]_{\widetilde\iota_{\delta}}
        \ar[r]^{\exp 2\pi\ii} & \UU(1) \ar[u]_{\iota_{\delta}}
\ar[r]^-{} & 1, }
$$
where $\psi(y) = \exp(2\pi\ii y/d)$, $y \in \RR$, and
$\varphi_{\delta} = \widetilde\iota_{\delta}
\circ\iota$. Note that for the
sequence of maps
$$
\xymatrix@!C{
\ZZ\ \ar[r]^-{\varphi_{\delta}} & \EE
        \ar[r]^-{\psi}  & \ S^1,}
$$
we have
\begin{equation} \label{eq:varphi}
\varphi_{\delta}(\ZZ) =
\psi^{-1}(\exp(2\pi\ii\delta/d)),
\end{equation}
which is a consequence of the exactness of (\ref{eq1}).

\subsection{Fibering Structures of Crystals in Continuum Picture}

Let us consider the $2$-dimensional euclidean space $\EE^2$
and some trivial bundles over $\EE^2$;
$\ZZ$-bundle $\pi_{\ZZ} : \ZZ_{\EE^2} \to \EE^2$,
$\EE$-bundle $\pi_{\EE} : \EE_{\EE^2} \to \EE^2$
and $S^1$-bundle $\pi_{S^1} : S^1_{\EE^2} \to \EE^2$.
We consider the bundle maps $\widehat\varphi_{\delta}$ and
$\widehat\psi$,
\begin{equation}
\xymatrix@!C{
 \ZZ_{\EE^2} \ar[r]^-{\widehat\varphi_{\delta}} &
 \EE_{\EE^2} \ar[r]^-{\widehat \psi} & S^1_{\EE^2}},
\label{eq:01}
\end{equation}
naturally induced by $\varphi_{\delta}$ and $\psi$,
respectively.

Note that
$\ZZ_{\EE^2} = \ZZ\times \EE^2$ is a covering space of $\EE^2$
and that $\EE_{\EE^2}$ is identified with $\EE^3 = \EE \times \EE^2$.
We sometimes use such identifications.

In the following, we often identify $\EE^2$ with the complex plane $\CC$.
For $\gamma \in S^1$,
let us consider the global constant section $\sigma_\gamma\in
\Gamma(\EE^2, S^1_{\EE^2})$ of $S^1_{\EE^2}$ defined by
$$
\sigma_\gamma(z) = (\gamma, z) \in S^1_{\EE^2} = S^1 \times \EE^2
$$
for $z \in \EE^2 = \CC$. The following lemma is straightforward
by virtue of (\ref{eq:varphi}).

\begin{lemma}\label{lm:2.1}
For $\gamma = \exp(2\pi\ii\delta/d)$,
we have
$$
\ZZ_{\EE^2, \gamma}=\widehat\varphi_{\delta}(\ZZ_{\EE^2}),
$$
where
$$
\ZZ_{\EE^2, \gamma}:=
{\widehat\psi^{-1}}
\left(
\sigma_{\gamma}(\EE^2)\right) \subset \EE^3.
$$
\end{lemma}

\subsection{Single Screw Dislocation in Continuum Picture}\label{subsec2.2}

For $z_0 \in \EE^2 = \CC$, let us consider the trivial bundles
$\EE_{\EE^2\setminus \{z_0\}}$ and
$S^1_{\EE^2\setminus \{z_0\}}$
over $\EE^2 \setminus \{z_0\}$
as in the previous subsection.
For $\gamma \in S^1$,
let us consider the section
$\sigma_{z_0, \gamma} \in
\Gamma(\EE^2\setminus \{z_0\}, S^1_{\EE^2\setminus \{z_0\}})$ defined by
$$
\sigma_{z_0, \gamma}(z) = \left(\gamma\frac{z-z_0}{|z-z_0|},z\right)
\mbox{ for } z \in \EE^2 \setminus \{z_0\} = \CC \setminus \{z_0\}.
$$
We set
$$
\ZZ_{\EE^2 \setminus \{z_0\}, \gamma}:=
{\widehat\psi^{-1}}
\left(\sigma_{z_0, \gamma}(\EE^2\setminus\{z_0\})\right) \subset
\EE_{\EE^2\setminus \{z_0\}} \subset \EE^3
$$
and let $\pi_{z_0, \gamma}: \ZZ_{\EE^2 \setminus \{z_0\}, \gamma} \to
\EE^2\setminus \{z_0\}$ be defined by
$\pi_{z_0, \gamma}=
\pi_{\EE}|{\ZZ_{\EE^2 \setminus \{z_0\}, \gamma}}$.

\begin{lemma}\label{lemma:cover}
The map  $\pi_{z_0, \gamma}: \ZZ_{\EE^2 \setminus \{z_0\}, \gamma} \to
\EE^2\setminus \{z_0\}$
defines a universal covering of
$\EE^2\setminus \{z_0\}$.
\end{lemma}

\begin{proof}
Over each point of $\EE^2$, $\widehat{\psi}$ is a covering map
and it is trivial as a family of covering maps. Therefore,
we see that $\pi_{z_0, \gamma}$ defines a covering map.

Furthermore, $\ZZ_{\EE^2 \setminus \{z_0\}, \gamma}$
is path-wise connected. This is seen as follows. Let us take arbitrary
two points $x_1$ and $x_2$. Then $\pi_{z_0, \gamma}(x_1)$
and  $\pi_{z_0, \gamma}(x_2)$ are connected by a path
in $\EE^2 \setminus \{z_0\}$. By lifting such a path,
we see that $x_1$ is connected in  $\ZZ_{\EE^2 \setminus \{z_0\}, \gamma}$
to a point $x_1'$ such that $\pi_{z_0, \gamma}(x_1') =
\pi_{z_0, \gamma}(x_2)$, which we denote by $\bar{x}$.
Then, by lifting a loop based at $\bar{x}$
which turns around $z_0$ for an appropriate number of times, we see that
$x_1'$ is connected to $x_2$ in $\ZZ_{\EE^2 \setminus \{z_0\}, \gamma}$.
Therefore, $\ZZ_{\EE^2 \setminus \{z_0\}, \gamma}$
is path-wise connected.

Moreover, the lift of a loop $\ell$ in $\EE^2 \setminus \{z_0\}$
by $\pi_{z_0, \gamma}$ is a loop if and only if its
winding number around $z_0$ vanishes. This is because
the section $\sigma_{z_0, \gamma}$ over $\ell$ winds around
$S^1$ by the same number of times as it winds around $z_0$.
This implies that the action of $\pi_1(\EE^2 \setminus \{z_0\})$
on $\ZZ_{\EE^2 \setminus \{z_0\}, \gamma}$ is effective.
Therefore,  $\pi_{z_0, \gamma}$ is a universal covering.
This completes the proof.
\end{proof}

\begin{remark}{\rm{
It should be noted that
$\ZZ_{\EE^2 \setminus \{z_0\}, \gamma}$ is embedded in
$\EE_{\EE^2\setminus \{z_0\}} \subset \EE^3$.
This represents a screw dislocation in a
crystal in continuum picture.
The point $z_0$ corresponds to the position of the dislocation line.
}}
\end{remark}

\begin{remark}\label{rmk:2.3}
{\rm{
A standard universal covering space
$\ZZ_{\EE^2\setminus\{z_0\}}$ of $\EE^2\setminus\{z_0\}$
is constructed as follows.
Fixing a point $x_0 \in \EE^2 \setminus \{z_0\}$,
we consider the path space (see \cite{BT})
$$
\Path(\EE^2\setminus \{z_0\}):=\{\mbox{continuous maps }\mu:[0,1]\to
\EE^2\setminus \{z_0\} \, | \, \mu(0) = x_0\}/\sim,
$$
where two paths $\mu$ and $\nu$ are equivalent, written
as $\mu \sim \nu$,
if $\mu(1)=\nu(1)$ and $\mu$ is
homotopic to $\nu$ relative to end points in
$\EE^2 \setminus \{z_0\}$.
The path space is the quotient space with respect to the
equivalence, where it is endowed with the natural topology
induced from that of the locally simply connected space
$\EE^2 \setminus \{z_0\}$.
Then
the map $\pi_\Path : \Path(\EE^2\setminus \{z_0\})\ni\mu \mapsto
\mu(1)\in  \EE^2\setminus \{z_0\}$
defines a universal covering. Note that
$\pi_\Path^{-1}(z)$, $z \in \EE^2 \setminus \{z_0\}$,
is regarded as the set of winding numbers
around $z_0$, i.e.\
$\pi_\Path^{-1}(z) = \ZZ$ up to a shift.
Then, we define
$$
\ZZ_{\EE^2\setminus\{z_0\}} := \Path(\EE^2\setminus \{z_0\}).
$$
It should be noted that the space obtained does not depend on the choice of the
point $x_0$, up to a covering equivalence.

By virtue of Lemma~\ref{lemma:cover} together with the
uniqueness of the universal covering,
we can construct an embedding
$\widehat \varphi_{z_0,\delta}:
\ZZ_{\EE^2\setminus\{z_0\}}\to \EE_{\EE^2\setminus\{z_0\}}$ such that
the diagram
$$
\xymatrix{
\ZZ_{\EE^2\setminus\{z_0\}}  \ar[rd]_{\pi_\Path}
\ar[rr]^{\widehat \varphi_{z_0,\delta}} & &
\EE_{\EE^2\setminus\{z_0\}} \ar[ld]^{\pi_\EE|\EE_{\EE^2 \setminus \{z_0\}}} \\
& \EE^2\setminus\{z_0\} &
}
$$
commutes and
$$
\widehat \varphi_{z_0,\delta}(\ZZ_{\EE^2\setminus\{z_0\}})
=
\ZZ_{\EE^2 \setminus \{z_0\}, \gamma} \subset \EE^3
$$
holds for $\gamma = \exp(2\pi\ii\delta/d)$.
In other words, we have the sequence
\begin{equation}
\xymatrix{
 \ZZ_{\EE^2\setminus\{z_0\}} \quad
 \ar[r]^-{\widehat \varphi_{z_0,\delta}} & \quad
 \EE_{\EE^2\setminus\{z_0\}} \quad
 \ar[r]^-{\widehat \psi} &  \quad S^1_{\EE^2\setminus\{z_0\}} }
\label{eq:02}
\end{equation}
as a non-trivial analogue of (\ref{eq:01}) in such a way that
$\widehat \psi \circ \widehat \varphi_{z_0,\delta}
(\ZZ_{\EE^2\setminus\{z_0\}}) = \sigma_{z_0, \gamma}(\EE^2
\setminus \{z_0\})$
and that $\widehat \varphi_{z_0,\delta}
(\ZZ_{\EE^2\setminus\{z_0\}})$ represents a
single screw dislocation in a crystal.
}}
\end{remark}

\subsection{Multi-Screw Dislocation in Continuum Picture}\label{subsec:multi}

Let us now consider multiple screw dislocations that are
parallel to each other. They are described as follows.

Let $\cS=\cS_+\coprod\cS_-$ be a finite subset of $\EE^2$, which
is divided into disjoint subsets $\cS_+$ and $\cS_-$.
Let us consider the trivial bundles
$\EE_{\EE^2\setminus \cS}$ and
$S^1_{\EE^2\setminus \cS}$
over $\EE^2 \setminus \cS$
as in the previous subsections. Then, we consider the section
$\sigma_{\cS, \gamma} \in
\Gamma(\EE^2\setminus \cS, S^1_{\EE^2 \setminus \cS})$
defined by
\begin{equation}
\sigma_{\cS, \gamma}(z) =
\left(\gamma
\prod_{z_i\in \cS_+}\frac{z-z_i}{|z-z_i|}\cdot
\prod_{z_j\in \cS_-}\frac{\overline{z-z_j}}{|z-z_j|},
z\right) \mbox{ \rm for }
z \in \EE^2 \setminus \cS = \CC \setminus \cS,
\label{eq:MDC}
\end{equation}
where
$\overline{z-z_j}$ is the complex conjugate of $z - z_j$.
We enumerate the points in $\cS$ in such a way that
$$\cS_+ = \{z_1, z_2, \ldots, z_s\}, \quad
\cS_- = \{z_{s+1}, z_{s+2}, \ldots, z_{s+t}\},$$
where $n = s+t$ is the cardinality of $\cS$.

\begin{definition}
\label{def:MDC}
Set
$$
\ZZ_{\EE^2 \setminus \cS, \gamma}:=
{\widehat\psi^{-1}}
\left(\sigma_{\cS, \gamma}(\EE^2\setminus\cS)\right) \subset
\EE^3\setminus
\pi_{\EE}^{-1}(\cS) \subset \EE^3
$$
and define $\pi_{\cS, \gamma}: \ZZ_{\EE^2 \setminus \cS, \gamma}\to
\EE^2\setminus \cS$
by $\pi_{\cS, \gamma}=\pi_{\EE}|{\ZZ_{\EE^2 \setminus \cS,\gamma}}$.
\end{definition}

As in Lemma~\ref{lemma:cover},
we see that $\pi_{\cS, \gamma}$ defines a covering map.
Note also that the covering space $\ZZ_{\EE^2 \setminus \cS, \gamma}$
of $\EE^2\setminus \cS$ is realized in $\EE^3$. This represents multiple
parallel screw dislocations in a crystal in continuum picture. The set $\cS$
corresponds to the positions of the dislocation lines.

Now, let us clarify the nature of the covering
$\pi_{\cS, \gamma}: \ZZ_{\EE^2 \setminus \cS, \gamma}\to
\EE^2\setminus \cS$.
As in Remark~\ref{rmk:2.3},
we can define the path space
$$
\Path(\EE^2\setminus \cS):=\{\mbox{continuous maps }\mu:[0,1]\to
\EE^2\setminus \cS\, |\, \mu(0) = x_0\}/\sim,
$$
using the same equivalence relation, where we
fix a point $x_0 \in \EE^2 \setminus \cS$. This is the
universal covering space of $\EE^2 \setminus \cS$.
However, for our purpose, this space is too big, and we
need to take certain quotients.

Note that the fundamental group $G = \pi_1(\EE^2 \setminus \cS, x_0)$
is a free group of rank $n$ generated by $m_1, m_2, \ldots, m_n$,
where $m_i$ is the element of $G$ represented by a loop based at $x_0$
which turns around $z_i$ once in the counterclockwise direction
and which does not turn around $z_j$, $j \neq i$.
Let $G' = [G, G]$ be the commutator subgroup of $G$.
By taking the orbit space under the natural action of $G'$,
we get the universal abelian covering space
of $\EE \setminus \cS$, denoted by
$\Path^a(\EE^2\setminus \cS)$.
In other words, we have
$$
\Path^a(\EE^2\setminus \cS):=\{\mbox{continuous maps }\mu:[0,1]\to
\EE^2\setminus \cS\, |\, \mu(0) = x_0\}/\sim_a,
$$
where for two paths $\mu$ and $\nu$, we have
$\mu \sim_a \nu$ if $\mu(1) = \nu(1)$
and the loop $\mu \ast \overline{\nu}$ based at
$x_0$ represents an element of $G'$,
where $\mu \ast \overline{\nu}$ denotes the
product path of $\mu$ and the inverse path of $\nu$.

Note that the homology group $H_1(\EE^2\setminus\cS; \ZZ)$
is isomorphic to $\ZZ^n$, which is freely generated by the homology classes
$[m_1], [m_2], \ldots, [m_n]$ represented by
$m_1, m_2, \ldots, m_n$, respectively. Note also that $H_1(\EE^2\setminus\cS; \ZZ)$
is isomorphic to the quotient group $G/G'$.
An arbitrary element $\kappa$ of $H_1(\EE^2\setminus\cS; \ZZ)$ is
represented as $\sum_i w_i [m_i]$, where
$w_i \in \ZZ$ is the winding number of $\kappa$ around $z_i$
in the direction of $m_i$. Therefore, for a loop $\ell$
in $\EE^2 \setminus \cS$, its lift in $\Path^a(\EE^2\setminus \cS)$
is a loop if and only if the winding number of $\ell$ around
each point of $\cS$ vanishes.

Let us now define
$$
 \ZZ^n_{\EE^2\setminus \cS}:= \Path^a(\EE^2\setminus \cS),
\quad
\pi_{\Path^a} :\ZZ^n_{\EE^2\setminus \cS}\to \EE^2\setminus \cS,
$$
where $\pi_{\Path^a}$ sends the class of each path to its terminal point.
Note that $\pi_{\Path^a}$ defines a covering, and that
for each $z \in \EE^2 \setminus \cS$, $\pi_{\Path^a}^{-1}(z)$
can be identified with $\ZZ^n$ up to a certain ``shift''.

Let us now take further quotients.
Let $h : \pi_1(\EE^2 \setminus \cS, x_0) \to \ZZ$
be the homomorphism defined by
$$
h(m_i)=\left\{
\begin{array}{lr}
1,& 1 \leq i \leq s,\\
-1,& s+1 \leq i \leq s+t.
\end{array}\right.
$$
Since $\ZZ$ is abelian, the kernel $H$ of $h$ contains $G'$.
Let $\ZZ_{\EE^2\setminus \cS}$ be the orbit space
of $\Path(\EE^2\setminus \cS)$
under the natural action of $H$.
In other words, we have
$$
\ZZ_{\EE^2\setminus \cS} := \{\mbox{continuous maps }\mu:[0,1]\to
\EE^2\setminus \cS\, |\, \mu(0) = x_0\}/\sim_h,
$$
where for two paths $\mu$ and $\nu$, we have
$\mu \sim_h \nu$ if $\mu(1) = \nu(1)$
and the loop $\mu \ast \overline{\nu}$ based at
$x_0$ represents an element of $H$.
We have a natural projection $\pi_H : \ZZ_{\EE^2\setminus \cS}
\to \EE^2 \setminus \cS$ which sends the class
of each path to its terminal point. Note that $\pi_H$ defines
a covering map.

\begin{proposition}\label{prop:2.6}
The coverings $\pi_{\cS, \gamma} : \ZZ_{\EE^2
\setminus \cS, \gamma} \to \EE^2 \setminus \cS$ and
$\pi_H : \ZZ_{\EE^2\setminus \cS}
\to \EE^2 \setminus \cS$ are equivalent.
\end{proposition}

\begin{proof}
Fix a point $\tilde{x}_0 \in \ZZ_{\EE^2
\setminus \cS, \gamma}$ such that $\pi_{\cS, \gamma}(\tilde{x}_0)
= x_0$. We define the map $\Phi : \ZZ_{\EE^2
\setminus \cS, \gamma} \to \ZZ_{\EE^2\setminus \cS}$
as follows. For $x \in  \ZZ_{\EE^2
\setminus \cS, \gamma}$, since $\ZZ_{\EE^2
\setminus \cS, \gamma}$ is connected, we can find
a path $\tilde{\mu}$ connecting $\tilde{x}_0$ and $x$.
Then, the composition $\mu = \pi_{\cS, \gamma} \circ \tilde{\mu}$
can be regarded as an element of $\ZZ_{\EE^2\setminus \cS}$
as a continuous path. We see that this map $\Phi$ is well-defined and injective,
by observing that the lift of a loop $\lambda$ based at $x_0$
in $\EE^2 \setminus \cS$ with respect to $\pi_{\cS, \gamma}$
is again a loop if and only if $h([\lambda])$ vanishes,
where $[\lambda] \in \pi_1(\EE^2 \setminus \cS, x_0)$
is the class represented by $\lambda$.
Furthermore, by the existence of a lift for paths,
we see that $\Phi$ is surjective. Since, we see easily that $\Phi$
is a local homeomorphism, we conclude that the map $\Phi$
gives the desired equivalence of coverings.
\end{proof}

As a consequence, we can construct an embedding
$\widehat{\varphi}_{\cS, \delta} : \ZZ_{\EE^2 \setminus \cS}
\to \EE_{\EE^2 \setminus \cS}$ such that
the diagram
$$
\xymatrix{
\ZZ_{\EE^2\setminus \cS}  \ar[rd]_{\pi_H}
\ar[rr]^{\widehat \varphi_{\cS,\delta}} & &
\EE_{\EE^2\setminus \cS} \ar[ld]^{\pi_\EE|\EE_{\EE^2 \setminus \cS}} \\
& \EE^2\setminus \cS &
}
$$
commutes and
$$
\widehat \varphi_{\cS,\delta}(\ZZ_{\EE^2\setminus\cS})
=
\ZZ_{\EE^2 \setminus \cS, \gamma} \subset \EE^3
$$
holds for $\gamma = \exp(2\pi\ii\delta/d)$.
In other words, we have the sequence
\begin{equation}
\xymatrix{
 \ZZ_{\EE^2\setminus\cS} \quad
 \ar[r]^-{\widehat \varphi_{\cS,\delta}} & \quad
 \EE_{\EE^2\setminus\cS} \quad
 \ar[r]^-{\widehat \psi} &  \quad S^1_{\EE^2\setminus\cS} }
\label{eq:02-2}
\end{equation}
as a non-trivial analogue of (\ref{eq:01}) such that
$\widehat \psi \circ \widehat \varphi_{\cS,\delta}
(\ZZ_{\EE^2\setminus\cS}) = \sigma_{\cS, \gamma}(\EE^2
\setminus \cS)$ and that $\widehat \varphi_{\cS,\delta}
(\ZZ_{\EE^2\setminus\cS})$ represents a multiple parallel
screw dislocations in a crystal, which generalizes (\ref{eq:02}).


In the following sections, we will consider discrete pictures
of dislocations, which will be embedded in these continuum pictures.

\section{Abelian Group Structure of SC Lattice and its
Screw Dislocations}\label{section3}

\subsection{Algebraic Structure of SC lattice}
The simple cubic (SC) lattice is usually expressed by
the additive free abelian group
$$
\AA_3^a := \ZZ a_1 + \ZZ a_2 + \ZZ a_3
= \LA a_1, a_2, a_3\RA_\ZZ,
$$
which is generated by three elements $a_1$, $a_2$ and $a_3$.
Here, by identifying
$a_1$ with $(a,0,0)$, $a_2$ with $(0,a,0)$, and $a_3$
with $(0,0,a)$ in $\RR^3$ endowed
with the euclidean inner product,
we see that $\AA^a_3$ is naturally embedded in $\RR^3$ as a SC lattice
as shown in Figure~\ref{fig:SC01} (a). We denote this embedding by
$\iota_{\AA_3} : \AA^a_3\hookrightarrow \RR^3$.
Furthermore, we embed $\AA^a_3$ into the euclidean $3$-space $\EE^3$
as follows.
By fixing $\delta  = (\delta_1, \delta_2, \delta_3) \in \EE^3$
and $g \in \SO(3)$, we define
$\iota_{\AA_3, \delta, g}: \AA^a_3 \hookrightarrow \EE^3$
by
$$
\iota_{\AA_3,\delta,g} (n_1 a_1+n_2 a_2 + n_3 a_3) =
g(n_1 a, n_2 a,  n_3 a) + (\delta_1,\delta_2,\delta_3)
$$
for $n_1, n_2, n_3 \in \ZZ$.
When $g\in \SO(3)$ is the identity,
we denote it by
$\iota_{\AA_3, \delta}: \AA^a_3 \hookrightarrow \EE^3$, which
is defined by
\begin{equation}
\iota_{\AA_3,\delta} (n_1 a_1+n_2 a_2 + n_3 a_3) =
(n_1 a, n_2 a,  n_3 a)  +  (\delta_1,\delta_2,\delta_3).
\label{eq:iota_A3}
\end{equation}

\begin{figure}[t]
 \begin{center}
\includegraphics[width=11cm]{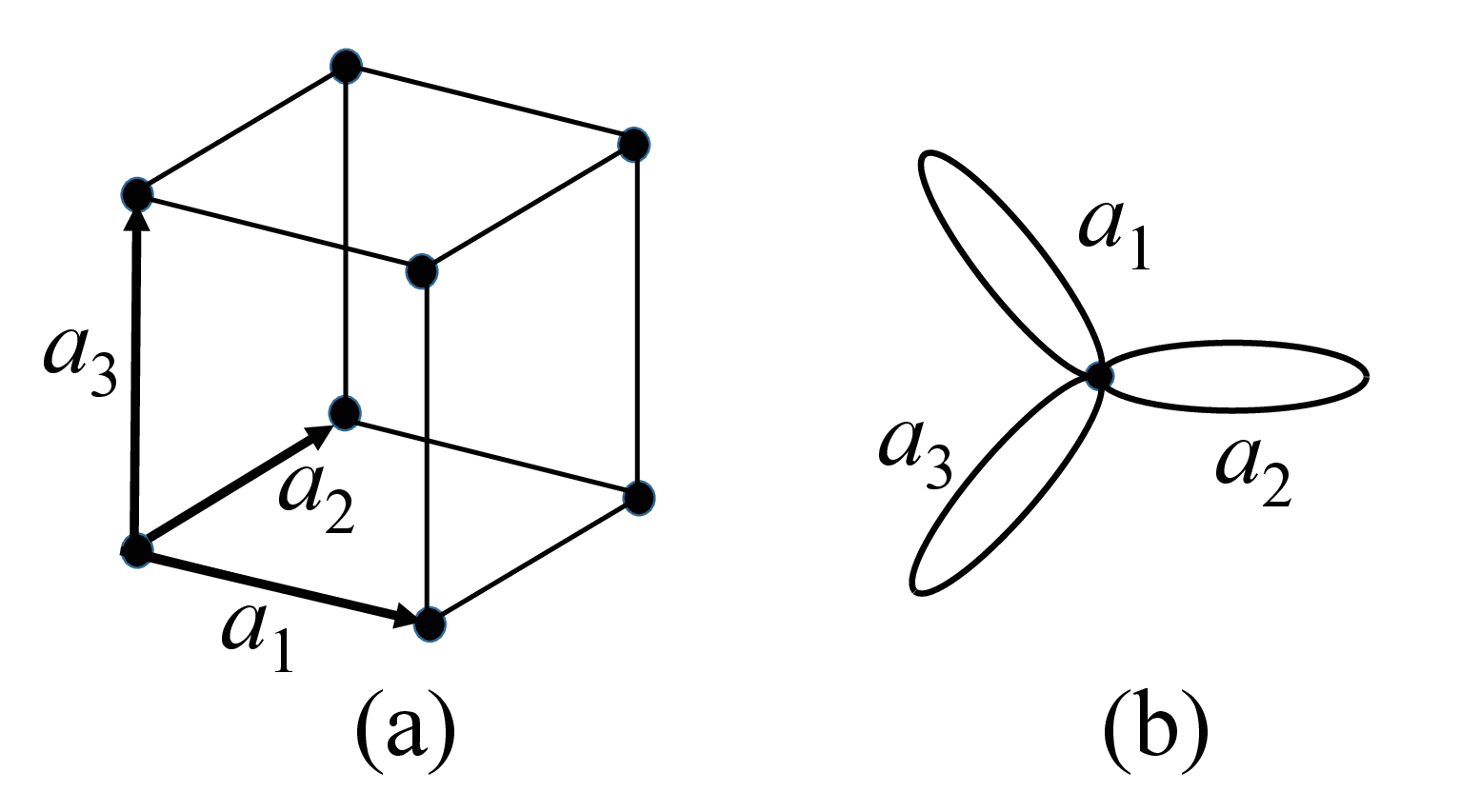}
 \newline
 \end{center}
\caption{
Simple cubic (SC) lattice}
\label{fig:SC01}
\end{figure}

\begin{remark}
{\rm{
We have the natural action of
the discrete subgroup of $\SO(3)$ on
the SC lattice \cite{CS,I,K,S},
which is isomorphic to the symmetric group $S_4$
on a set of four elements \cite{I},
although it does not play an important role
in this article.
}}
\end{remark}

In this article, we often use the multiplicative abelian group
$$
\AA_3 := \{\alpha_1^{n_1}\alpha_2^{n_2}\alpha_3^{n_3} \,
| \, \mbox{abelian}, n_1, n_2, n_3 \in \ZZ\}
$$
rather than the additive group $\AA_3^a$ for convenience.
Here, $\alpha_b$ corresponds to $a_b$, $b=1,2,3$.
Through the natural identification of $\AA_3$ with
$\AA_3^a$, we continue to use
the symbols $\iota_{\AA_3}$ and
$\iota_{\AA_3,\delta}$ also for $\AA_3$
by abuse of notation.

\begin{remark}
{\rm{
The SC lattice is regarded as the universal abelian covering of a certain
geometric object.
In fact, the SC lattice is expressed by the graph as depicted in
Figure~\ref{fig:SC01} (b) \cite{S}.
More precisely, the universal abelian covering
of the graph given by Figure~\ref{fig:SC01} (b)
coincides with the Cayley graph of the group $\AA_3$ with
respect to the generating set $\{\alpha_1, \alpha_2, \alpha_3\}$.
}}
\end{remark}

\subsection{Fibering Structure of SC Lattice: $\revs{(0,0,1)}$-direction}
Let us introduce the group ring
$$
\CC[\AA_3] = \CC
[\alpha_1,\alpha_2, \alpha_3,\alpha_1^{-1},\alpha_2^{-1}, \alpha_3^{-1}]
$$
in order to consider
the fibering structure of $\AA_3$ (and that of
$\AA_3^a$).
Its projection to the $2$-dimensional space corresponds to
taking the quotient as
$$
\CC[\AA_3]/(\alpha_3-1) = \CC[\alpha_1,\alpha_2, \alpha_1^{-1},\alpha_2^{-1}
] =: \CC[\AA_2],
$$
where
$$
\AA_2 := \{\alpha_1^{n_1}\alpha_2^{n_2} \,
| \, \mbox{abelian}, n_1, n_2 \in \ZZ\} ,
$$
which is group-isomorphic to
$\AA_2^a := \ZZ a_1 + \ZZ a_2$, and $(\alpha_3-1)$ is the ideal
generated by $\alpha_3-1$.
As in the case of $\AA_3$,
by assuming that $\alpha_1$ and $\alpha_2$ correspond
to $(a,0)$ and $(0,a)$ in $\RR^2$, respectively,
we may regard \revsSSS{$\AA_2^a \cong \AA_2$} as being also naturally embedded in $\RR^2$.
Thus, we have natural embeddings
$\iota_{\AA_2} :\AA_2 \hookrightarrow \RR^2$ and
$\iota_{\AA_2, \bar\delta} : \AA_2 \hookrightarrow \EE^2$
\revsSSS{defined by $\iota_{\AA_2, \bar\delta}(x) = \iota_{\AA_2}(x) + \bar \delta$}
for $\bar\delta = (\delta_1, \delta_2) \in \EE^2$.

The projection above ``induces'' the fibering structure
$$
\xymatrix{
 \ZZ  \ar[r] & \AA_3 \ar[d]^\varpi \\
 &  \AA_2,
}
$$
where $\varpi$ is the projection defined by $\varpi(\alpha_1^{n_1}\alpha_2^{n_2}
\alpha_3^{n_3}) = \alpha_1^{n_1}\alpha_2^{n_2}$, $n_1, n_2, n_3 \in \ZZ$.
Its graph expression is given by Figure~\ref{fig:SC001}:
more precisely, the Cayley graph of $\AA_2$ with
respect to the generating set $\{\alpha_1, \alpha_2,
\alpha_3\}$ coincides with the graph depicted in Figure~\ref{fig:SC001}.
We will see that the above fibering structure is essential
in our description of screw dislocations.

\begin{figure}[t]
 \begin{center}
\includegraphics[width=11cm]{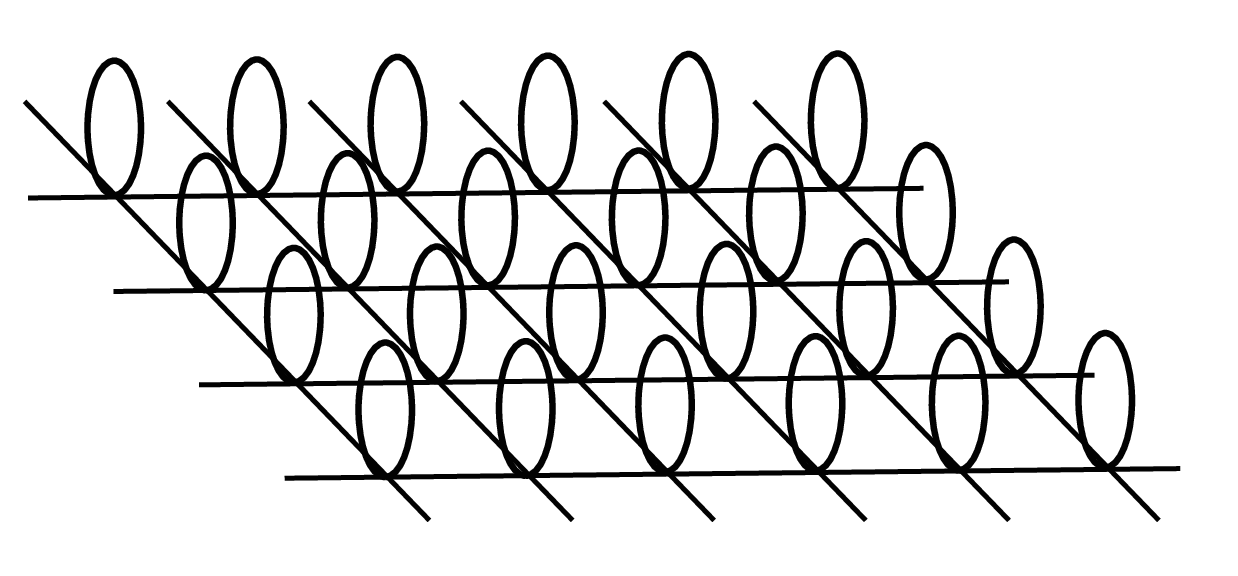}
 \newline
 \end{center}
\caption{
Fibering structure of simple cubic lattice}
\label{fig:SC001}
\end{figure}

Note that we may regard the group ring
$\CC[\AA_2]$ as a set of certain complex valued functions on
$\AA_2$: for an element of $f \in \CC[\AA_2]$, we have the
complex number $f(n_1, n_2) \in \CC$
for every element $\alpha_1^{n_1}\alpha_2^{n_2}
\in \AA_2$ in such a way that we have
$$
f = \sum_{(n_1, n_2)\in \ZZ^2}
 f(n_1, n_2) \alpha_1^{n_1}\alpha_2^{n_2}.
$$
In this case, $f$ can take non-zero complex values
only on a finite number of elements of $\AA_2$.
We extend this space to the whole function space $\cF(\fLSCp, \CC)$,
where we use the symbol $\fLSCp$ for $\AA_2$ viewed as a set
or as a discrete topological space, i.e.,
$$
    \fLSCp:=\{
       (n_1 a, n_2 a)\, | \, n_1, n_2 \in \ZZ\}.
$$

We denote by $S^1_{\fLSCp}$ the trivial $S^1$-bundle
over $\fLSCp$.
For every fiber bundle $F_{\EE^2}$ over $\EE^2$
and for $\bar\delta = (\delta_1, \delta_2) \in \EE^2$,
by the natural embedding $\iota^{\SC}_{\bar\delta}
: \fLSCp \hookrightarrow \EE^2$ defined by
$(n_1 a, n_2 a) \mapsto (n_1 a +\delta_1, n_2 a +\delta_2)$,
we have the pullback bundle $F_{\fLSCp}$ over $\fLSCp$.

Let $\cS=\cS_+\coprod\cS_-$ be a finite subset  in $\EE^2$
as in Subsection~\ref{subsec:multi}.
In the following, we assume that $\bar\delta \in \EE^2$ satisfies
$\iota^{\SC}_{\bar\delta}(\fLSCp)
\cap \cS= \emptyset$. Then,
the map $\iota^{\SC}_{\bar\delta}$ is regarded as the embedding
$\iota^\SC_{\bar\delta} : \fLSCp
\hookrightarrow \EE^2\setminus \cS$.
Thus, we have the following.

\begin{lemma}
For every fiber bundle $F_{\EE^2\setminus \cS}$
over $\EE^2 \setminus \cS$,
by the embedding
$\iota^\SC_{\bar\delta} : \fLSCp \hookrightarrow
\EE^2\setminus \cS$,
we have the pullback bundle $F_{\fLSCp}$
that satisfies the
commutative diagram
$$
\xymatrix{
 F_{\fLSCp} \ar[d] \ar[r]^{\hat\iota^\SC_{\bar\delta}}
& F_{\EE^2\setminus \cS}  \ar[d]
\\
 \fLSCp  \ar[r]^{\iota^\SC_{\bar\delta}} & \EE^2\setminus \cS,
}
$$
where the vertical maps are the projections
of the fiber bundles and $\hat\iota^\SC_{\bar\delta}$
is the bundle map induced by $\iota^\SC_{\bar\delta}$.
\end{lemma}

Recall that in Section~\ref{section2}, we have fixed $d > 0$.
In the following, we set $d = a$.
Using the above pullback diagram (Cartesian square), we have the following.

\begin{lemma} \label{lm:SC001}
We have the following commutative diagram:
$$
\xymatrix{\EE_{\fLSCp} \ar[r]^{\hat\iota^\SC_{\bar\delta}}
\ar[d]
\ar@/_24pt/[dd]_{\widehat{\psi}}
& \EE_{\EE^2\setminus \cS}
\ar[d]
\ar@/^24pt/[dd]^{\widehat{\psi}} \\
\fLSCp \ar[r]^{\iota^\SC_{\bar\delta}} &
\EE^2 \setminus \cS \\
S^1_{\fLSCp} \ar[u] \ar[r]^{\hat\iota^\SC_{\bar\delta}} &
S^1_{\EE^2 \setminus \cS},
\ar[u]
}
$$
where the straight vertical arrows are projections
of the fiber bundles and $\widehat{\psi}$ are
the bundle maps induced by $\psi$ defined in Section~\textup{\ref{section2}}.
\end{lemma}

The above lemma is easy to prove, and hence we
omit the proof.

The following proposition corresponds
to the case where $\cS = \emptyset$ and the proof
is left to the reader.

\begin{proposition}
Set $\gamma=\exp(2\pi\ii \delta_3/a)
\in S^1$ for a $\delta_3 \in \RR$ and consider the
global section $\check\sigma_\gamma \in
\Gamma(\fLSCp, S^1_{\fLSCp})$ that constantly takes
the value $\gamma$.
Then, we have that
$$
\hat\iota^\SC_{\bar \delta}
\left( \widehat\psi^{-1}\left(\check\sigma_\gamma(\fLSCp)\right)\right)
=\hat\iota^\SC_{\bar \delta}\left(
        \frac{a}{2\pi\ii} \exp^{-1}\left(\check\sigma_\gamma(\fLSCp)
          \right)\right)
     \subset \EE_{\EE^2 \setminus \cS}
     \subset \EE^3
$$
coincides with $\iota_{\AA_3,\delta}(\AA^a_3)$ as a subset in $\EE^3$
for $\delta = (\bar \delta, \delta_3) = (\delta_1, \delta_2, \delta_3)$, i.e.,
$$
\iota_{\AA_3,\delta}(\AA^a_3)=
\hat\iota^\SC_{\bar \delta}
\left( \widehat\psi^{-1}
\left(\check\sigma_\gamma(\fLSCp)\right)\right).
$$
\end{proposition}

In other words, the SC lattice without dislocation can be
interpreted as the inverse image by $\widehat{\psi}$
of a constant section $\check\sigma_\gamma$ of the trivial $S^1$-bundle.


\subsection{Screw Dislocation in Simple Cubic Lattice}\label{subsec:sdSC}

A screw dislocation in the simple cubic lattice appears along
the $\revs{(0,0,1)}$-direction \cite{N} \insMMM{
up to automorphisms of the SC lattice.
In other words, we may assume that the Burgers vector is parallel to
the $(0, 0, 1)$-direction.}

Using the fibering structure of
Lemma~\ref{lm:SC001}, for the case where
$\cS = \{z_0\}$, $z_0 \in \CC$,
we can describe a single screw dislocation in the SC lattice
as follows, whose proof is straightforward.
Our principal idea is to use a section in $\Gamma(\fLSCp, S^1_{\fLSCp})$
in order to describe a screw dislocation.
\revsSSS{In the following, we set $z_0' = z_0 - (\delta_1 + \delta_2 \ii)$,
where $\bar \delta = (\delta_1, \delta_2)$.}

\begin{proposition}\label{prop:SC-single}
Let us define the section $\check \sigma_{z_0', \gamma}
\in \Gamma(\fLSCp, S^1_{\fLSCp})$ by \revsSSS{
$$
     \check \sigma_{z_0',\gamma}(n_1 a,n_2 a)
          = \left(\gamma\frac{(n_1 a + n_2 a \ii)  - z_0'}
                       {|(n_1 a + n_2 a \ii)  - z_0'|},
                         (n_1 a, n_2 a) \right),
     \quad (n_1 a, n_2 a) \in \fLSCp,
$$}
where
$\gamma=\exp(2\pi \ii \delta_3/a)
\in S^1$.
Then, the screw dislocation around $z_0$
given by
$$
\fD_{z_0}^\SC
:= \hat\iota^\SC_{\bar \delta}\left( \widehat{\psi}^{-1}
(\check \sigma_{z_0', \gamma}(\fLSCp))
\right) =
\hat\iota^\SC_{\bar \delta}\left( \frac{a}{2\pi\ii} \exp^{-1}
         \left(\check \sigma_{z_0', \gamma}(\fLSCp)\right)   \right)
$$
is realized in $\EE^3$.
\end{proposition}

We note that $\fD_{z_0}^\SC$ can be regarded as a kind of a ``covering space
of the lattice $\fLSCp$''.

\begin{remark}
{\rm{
(1)
Let $\ZZ_{\fLSCp, \gamma}$ be the pullback
of $\ZZ_{\EE^2\setminus \{z_0\}, \gamma}$ by
$\iota^\SC_{\bar \delta} : \fLSCp \to
\EE^2 \setminus \{z_0\}$, and
$\hat\iota^\SC_{\bar \delta}:\ZZ_{\fLSCp, \gamma}
\to \ZZ_{\EE^2\setminus \{z_0\}, \gamma}$
the induced bundle map. Then,
we have the equality
$$
\fD_{z_0}^\SC =
\hat\iota^\SC_{\bar \delta}
(\ZZ_{\fLSCp, \gamma}),
$$
which follows directly from the definition of
$\ZZ_{\EE^2\setminus \{z_0\}, \gamma}$
given in Subsection~\ref{subsec2.2}.

(2) Here, $\cS = \{z_0 \}$ corresponds to
the position of the dislocation line (see Figure~\ref{fig:SCdis01}).
As has been mentioned in the continuum model,
$\fD_{z_0}^\SC$ naturally embeds into the path space
$\ZZ_{\EE^2 \setminus \{z_0\}}$ via the
identification $\widehat \varphi_{z_0, \delta_3} : \ZZ_{\EE^2
\setminus \{z_0\}} \to \ZZ_{\EE^2\setminus \{z_0\}, \gamma}$.
}}
\end{remark}

\begin{figure}[t]
 \begin{center}
\includegraphics[width=10cm]{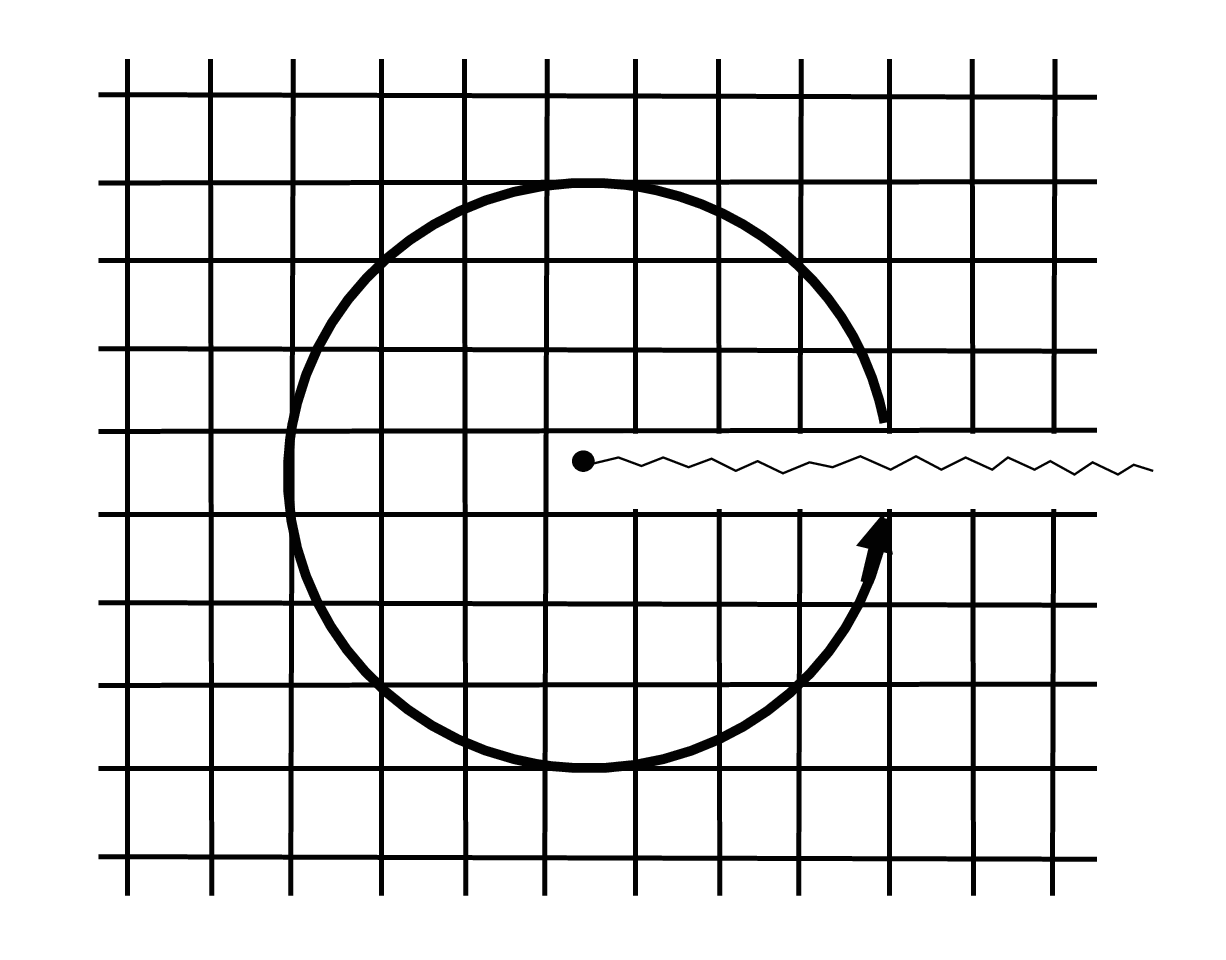}
 \newline
 \end{center}
\caption{Screw dislocation in the SC lattice: the lattice points
correspond to $\fLSCp \subset \EE^2$ and the center point is $z_0$.}
\label{fig:SCdis01}
\end{figure}
A parallel multi-screw dislocation is described as follows.
Let us define
the section $\check{\sigma}_{\cS', \gamma} \in \Gamma(\fLSCp, S^1_{\fLSCp})$ by
\revsSSS{
$$
\begin{array}{l}
 \check \sigma_{\cS',\gamma}(n_1 a,n_2 a) = \\
\displaystyle{\left(\gamma \prod_{z_i' \in \cS_+'}
        \frac{(n_1 a + n_2 a \ii) - z_i'}
        {|(n_1 a + n_2 a \ii) - z_i'|}
        \prod_{z_j' \in \cS_-'} \frac{\overline{(n_1 a + n_2 a \ii) - z_j'}}
                  {|(n_1 a + n_2 a \ii) - z_j'|},
                 (n_1 a, n_2 a) \right),} \\
     \hspace*{9cm} (n_1 a, n_2 a) \in \fLSCp,
\end{array}
$$
where
\begin{eqnarray*}
& & z_i' = z_i - (\delta_1 + \delta_2 \ii),\,
z_j' = z_j - (\delta_1 + \delta_2 \ii), \\
& & \cS_+' = \{z_i' \,|\, z_i \in \cS_+\}, \, 
\cS_-' = \{z_j' \,|\, z_j \in \cS_-\}
\end{eqnarray*}
and $\cS' = \cS_+' \coprod \cS_-'$.}
Then, we have the following, whose proof is straightforward.

\begin{proposition}\label{prop:SC-multi}
The parallel multi-screw dislocation
$\fD_{\cS}^\SC$ given by \revsSSS{
$$
\fD_{\cS}^\SC
=\hat\iota^\SC_{\bar \delta}\left( \widehat{\psi}^{-1}(\check \sigma_{\cS',\gamma}
  (\fLSCp))\right) = \hat \iota^\SC_{\bar \delta}\left(
\frac{a}{2\pi\ii} \exp^{-1}
(\check \sigma_{\cS',\gamma}(\fLSCp))\right)
$$}
is realized in $\EE^3$ as a subset of $\ZZ_{\EE^2\setminus \cS,\gamma}$,
where $\cS$ corresponds to
the position of the dislocation lines.
\end{proposition}

\section{Fibering Structure of BCC Lattice: $\revs{(1,1,1)}$-direction
and its Screw Dislocation} \label{section4}

\subsection{Fibering Structure of SC Lattice: $\revs{(1,1,1)}$-direction}
\label{subsec:SC111}

In this subsection, let us first consider the fibering structure along the
$\revs{(1,1,1)}$-direction of the simple cubic lattice.
Although a screw dislocation does not occur along
this direction physically, this construction is useful for analyzing the
case of the BCC lattice.

This structure is a little bit complicated; however, our algebraic
approach makes the computation easy and
enables us to have its discrete geometric interpretation.
Let us consider the projection which corresponds to
the quotient ring
$$
\RSCT =\CC[\AA_3]/(\alpha_1\alpha_2\alpha_3-1),
$$
where $\AA_3$ is the multiplicative free abelian group
of rank $3$ generated by $\alpha_1$, $\alpha_2$ and $\alpha_3$.

For the vector $a_1+a_2+a_3 \in \AA_3^a \subset \RR^3$,
we have the vanishing euclidean inner products
$$
(a_1-a_3, a_1+a_2+a_3)=0, \quad
(a_2-a_3, a_1+a_2+a_3)=0.
$$
Therefore, $a_1-a_3$ and $a_2-a_3$ constitute
a basis for the orthogonal complement of the vector
$a_1+a_2+a_3$ in $\RR^3$.
This space corresponds to
the group ring generated by $\alpha_1\alpha_3^{-1}$
and $\alpha_2\alpha_3^{-1}$.
In other words, we consider
$$
\AASCd:=\{
(\alpha_1\alpha_3^{-1})^{\ell_1}
(\alpha_2\alpha_3^{-1})^{\ell_2} \, | \,
\ell_1, \ell_2 \in \ZZ\},
$$
which is a subgroup of $\AA_3$.
We also consider the group ring
$$
\CC[\AASCd] =
\CC[\alpha_1 \alpha_3^{-1}, \alpha_2 \alpha_3^{-1},
\alpha_1^{-1} \alpha_3, \alpha_2^{-1} \alpha_3],
$$
which is identified, under the relation $\alpha_1 \alpha_2\alpha_3 =1$,
with
$$
\CC[\alpha_1 \alpha_2^{2}, \alpha_1^2 \alpha_2,
\alpha_1^{-1} \alpha_2^{-2}, \alpha_1^{-2} \alpha_2^{-1}].
$$
Note that, as $\CC[\AASCd]$ is a sub-ring of $\CC[\AA_3]$,
$\RSCT$ is also considered to be a $\CC[\AASCd]$-module.

\begin{lemma} \label{lm:SC111}
We have a natural isomorphism as
$\CC[\AASCd]$-modules:
$$
\RSCT
\cong
\CC[\AASCd]
\oplus\CC[\AASCd]\alpha_1
\oplus\CC[\AASCd]\alpha_1\alpha_2.
$$
\end{lemma}

\begin{proof}
First, note that every monomial of $\CC[\AA_3]$ has its
own degree with respect to $\alpha_1, \alpha_2$ and $\alpha_3$, each
of which has degree $1$.
Furthermore, it is easy to verify that
a monomial has degree zero if and only if it
belongs to $\CC[\AASCd]$.
As a result, a monomial has degree $r$
if and only if it belongs to $\CC[\AASCd] \alpha_1^r$,
where we can replace $\alpha_1^r$ by any other monomial
of degree $r$.
Note that each $\CC[\AASCd] \alpha_1^r$ is a $\CC[\AASCd]$-submodule
of $\CC[\AA_3]$.
Thus, we have the isomorphism
$$\CC[\AA_3] \cong \bigoplus_r \CC[\AASCd] \alpha_1^r$$
as $\CC[\AASCd]$-modules.

Now, let us consider the $\CC[\AASCd]$-module
homomorphism induced by the
inclusion
$$q : \bigoplus_{r=0}^2 \CC[\AASCd] \alpha_1^r \to
\left.\left(\bigoplus_r \CC[\AASCd] \alpha_1^r\right)\right/
(1 - \alpha_1 \alpha_2 \alpha_3)
\cong \CC[\AA_3]/(1 - \alpha_1 \alpha_2 \alpha_3),
$$
where the ideal
$(1 - \alpha_1 \alpha_2 \alpha_3)$ in the ring $\CC[\AA_3]$
is now considered as a $\CC[\AASCd]$-submodule.
We can easily show that this homomorphism $q$
is injective, since every non-zero element of
the submodule $(1 - \alpha_1\alpha_2\alpha_3)$
contains two non-zero monomials whose degrees
are different by a non-zero multiple of $3$.
Furthermore, $q$ is surjective,
since we have $\CC[\AASCd] \alpha_1^r
= \CC[\AASCd] \alpha_1^{r'}$ as long as $r \equiv r' \pmod{3}$,
under the relation $\alpha_1\alpha_2\alpha_3 = 1$.
Furthermore, we have $\CC[\AASCd] \alpha_1^2 =
\CC[\AASCd] \alpha_1 \alpha_2$. Thus, we have the
desired isomorphism of $\CC[\AASCd]$-modules.
This completes the proof.
\end{proof}

\begin{remark}
{\rm{
Lemma~\ref{lm:SC111}
can be geometrically interpreted by using a cube as follows.
%
Let us consider the cube
whose vertices consist of
$0$, $a_1$, $a_2$, $a_3$, $a_1+a_2$, $a_2+a_3$, $a_3+a_1$ and
$a_1+a_2+a_3$, which are identified with their corresponding elements
$1$, $\alpha_1$, $\alpha_2$, $\alpha_3$,
$\alpha_1\alpha_2$, $\alpha_2\alpha_3$, $\alpha_3\alpha_1$ and
$\alpha_1\alpha_2\alpha_3$, respectively.
We have the natural action of the
symmetric group on the three elements $\alpha_1$, $\alpha_2$ and
$\alpha_3$ over the cube, which can be regarded as
a subgroup of the octahedral group.
The orbits of this action are given by
$\{1\}$, $\{\alpha_1, \alpha_2, \alpha_3\}$,
$\{\alpha_1\alpha_2, \alpha_2\alpha_3, \alpha_3\alpha_1\}$
and $\{\alpha_1\alpha_2\alpha_3\}$, which coincide
with the classification by their
degrees. By identifying $1$ and $\alpha_1\alpha_2\alpha_3$, we get
the three classes as described in Lemma~\ref{lm:SC111}.
%
}}
\end{remark}

\insS{
\begin{remark} \label{rmk:Snake1}
{\rm{
We have another algebraic proof for Lemma~\ref{lm:SC111} as
follows.\footnote{The authors are indebted to an anonymous  referee for this
simple proof \insSS{as well as that described in Remark~\ref{rmk:Snake2}
for Lemma~\ref{lm:BCC}}.}
Let us consider the diagonal monomorphism $\iota_3^\diag :
\ZZ \hookrightarrow \AA_3^a$ defined by $\iota_3^\diag(n)
= n(a_1+a_2+a_3)\in \AA_3^a$ for $n \in \ZZ$, and denote
its image by $\AA_3^\diag$.
Then we naturally have
$\RSCT \cong \CC[\AA_3^a/\AA_3^\diag]$.
Furthermore, we also have the epimorphism $p_3:
\AA_3^a \to \ZZ$ defined by $p_3(n_1a_1+n_2a_2+n_3a_3) = n_1+n_2+n_3$ for
$n_1, n_2, n_3 \in \ZZ$.
Setting
$\AA_d^a=\ZZ(a_1-a_3)+\ZZ(a_2-a_3)=\Ker\ p_3$, we have
the commutative diagram of exact rows
\insSS{$$
\xymatrix{
& 0  \ar[r]\ar[d] &\ZZ \ar[r]\ar[d]^{\iota_3^\diag}  &3\ZZ \ar[r]\ar[d] &0\\
0 \ar[r] & \AA_d^a \ar[r] &\AA_3^a \ar[r]^{p_3} &\ZZ \ar[r] &0,
}
$$}where the rightmost vertical map is the natural inclusion.
By applying the snake lemma \cite{L} to this diagram, we obtain the short exact sequence
$$
\xymatrix{
 0  \ar[r]& \AA_d^a \ar[r] & \AA_3^a/\AA_3^\diag \ar[r]^{\overline{p_3}}
&  \ZZ/3\ZZ
 \ar[r]& 0,
}
$$
where $\overline{p_3}$ is the epimorphism
induced by $p_3$.
The inverse images of the three elements of $\ZZ/3\ZZ$
by $\overline{p_3}$ give the decomposition of
$\AA_3^a/\AA_3^\diag$ into three disjoint subsets,
$\AA_d^a$, $a_1 + \AA_d^a$ and $a_1+a_2 + \AA_d^a$.
This implies Lemma~\ref{lm:SC111}.
}}
\end{remark}
}

\begin{remark}
{\rm{
Lemma~\ref{lm:SC111} means
geometrically that the projection generates three sheets if we consider them
embedded in $\EE^3$ as shown in Figure~\ref{fig:SC111}.
Each sheet can be regarded as a set given by the abelian group
$\LA \alpha_1 \alpha_3^{-1}, \alpha_2 \alpha_3^{-1} \RA$.
More precisely,
\insS{let us denote \revsSSS{$\AA_3^a/\AA_3^\diag$} by $\fLSCd$ when considered as
\insSS{a set}. Then, we have the decomposition
$$
       \fLSCd:= \fLSCd^{(0)} \coprod \fLSCd^{(1)}\coprod \fLSCd^{(2)},
$$
where
\begin{eqnarray*}
       \fLSCd^{(0)}&:= & \{ \ell_1(a_1-a_3) +\ell_2(a_2-a_3) \,
                 | \, \ell_1, \ell_2 \in \ZZ\},\\
       \fLSCd^{(1)}&:= & \{ \ell_1(a_1-a_3) +\ell_2(a_2-a_3) + a_1 \,
                 | \, \ell_1, \ell_2 \in \ZZ\},\\
       \fLSCd^{(2)}&:= & \{ \ell_1(a_1-a_3) +\ell_2(a_2-a_3) + a_1 + a_2 \,
                 | \, \ell_1, \ell_2 \in \ZZ\}.
\end{eqnarray*}
Note that this corresponds
exactly to the decomposition mentioned in Remark~\ref{rmk:Snake1}.}
This picture comes from the discrete nature of the lattice.
The interval between the sheets is given by $\sqrt{3}a/3$.
}}
\end{remark}

\begin{figure}[t]
\begin{center}
\includegraphics[width=11cm]{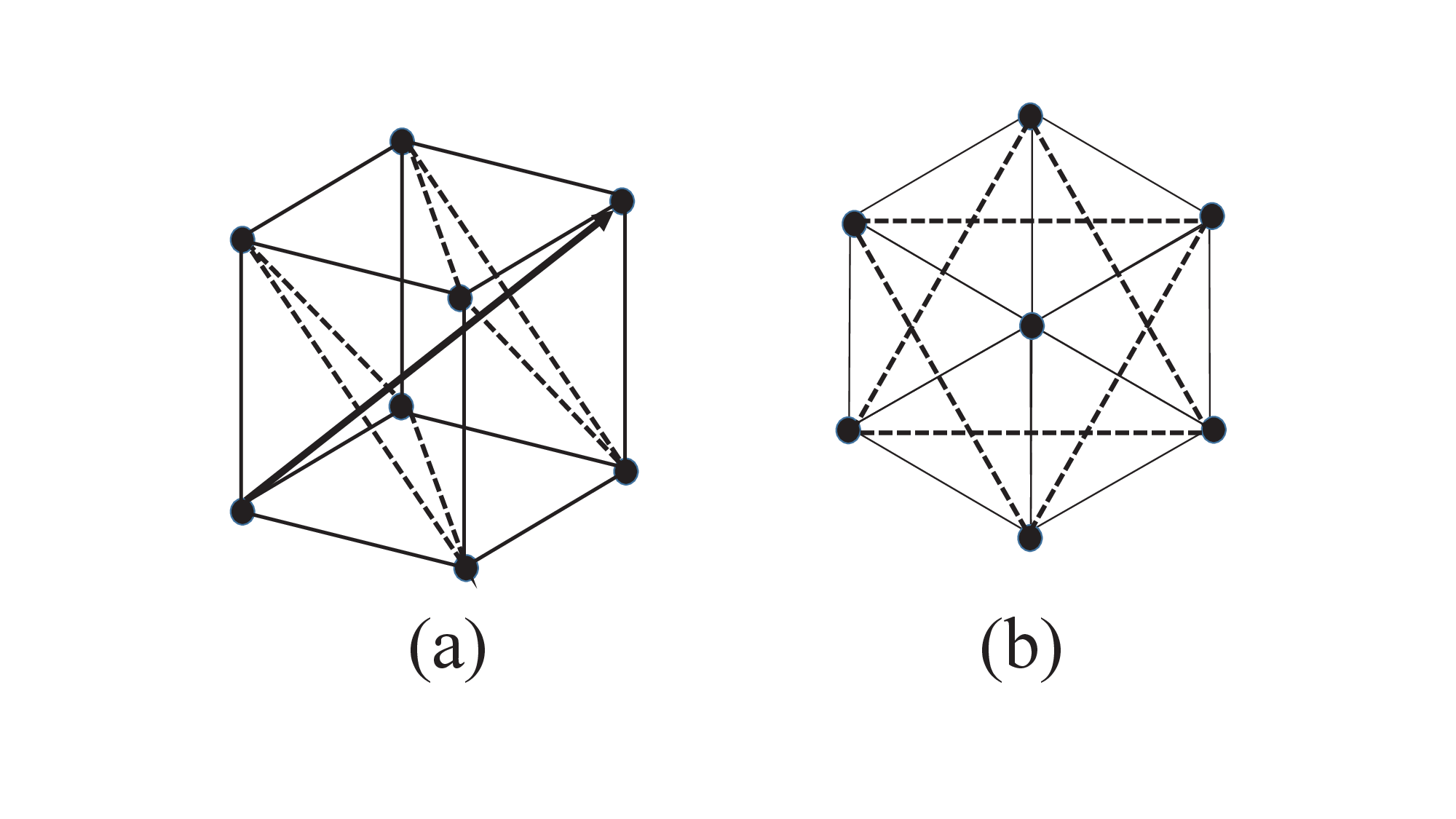}
%

\includegraphics[width=6.1cm]{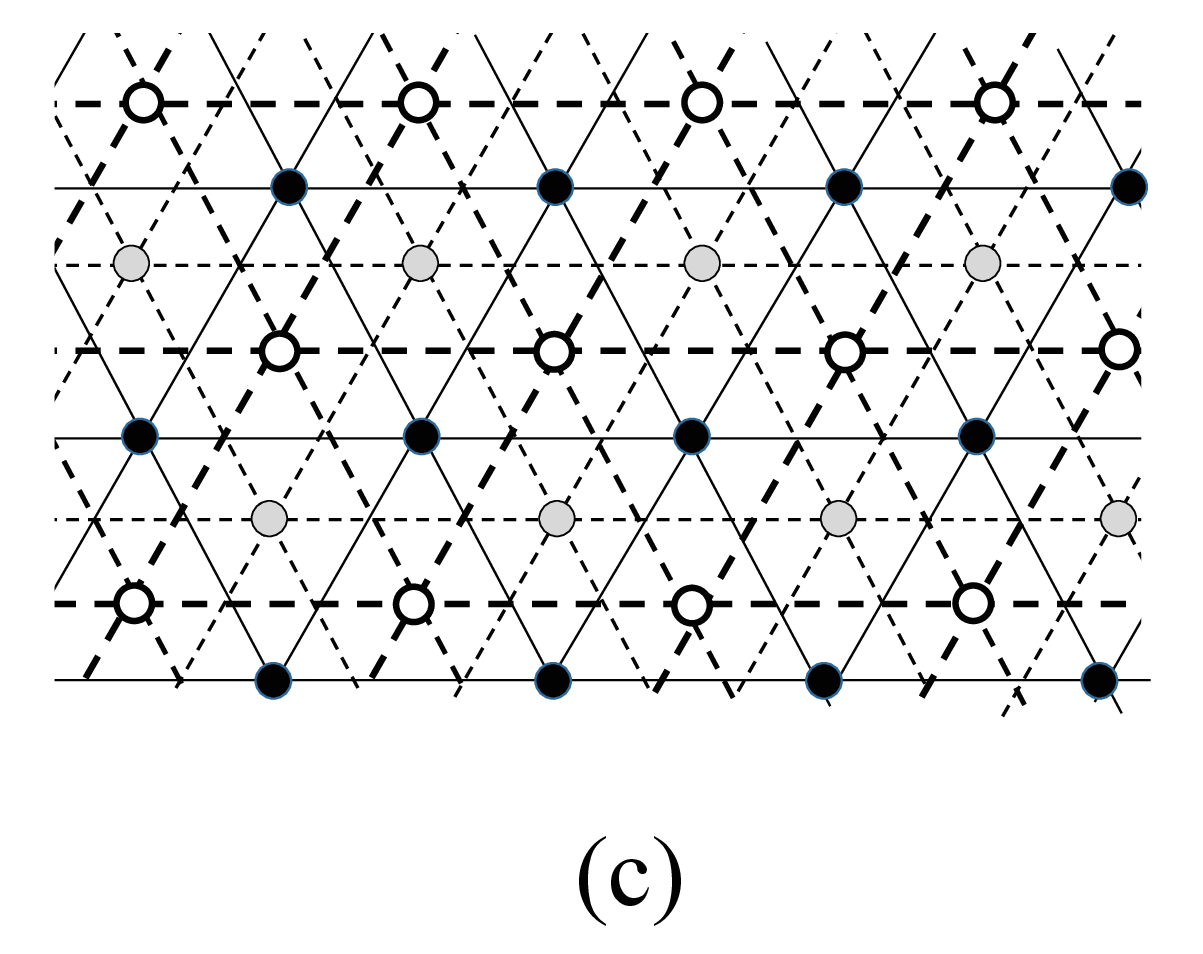}
\end{center}

\caption{The cube in (a) shows
the triangles whose normal direction is $\revs{(1,1,1)}$ in simple cubic lattice.
If one looks at the cube from the $\revs{(1,1,1)}$-direction, then the image is as in (b).
Furthermore, if one projects the whole SC lattice to the plane perpendicular
to the $\revs{(1,1,1)}$-direction, then one gets the image as in (c), where
the black, gray and white dots correspond to the three sheets
$\fLSCd^{(0)}$, $\fLSCd^{(1)}$ and $\fLSCd^{(2)}$, \revsSSS{respectively}.
}
\label{fig:SC111}
\end{figure}

\subsection{Algebraic Structure of BCC Lattice}

The BCC (body centered cubic) lattice is the lattice
in $\RR^3$ generated by $a_1$, $a_2$, $a_3$ and
$(a_1+a_2+a_3)/2$.
Algebraically, it
is described as
additive group (or $\ZZ$-module) by
$$
\AABCCd^a := \LA a_1, a_2, a_3, b\RA_\ZZ/\LA 2b-a_1-a_2-a_3 \RA_\ZZ,
$$
where $\LA 2b-a_1-a_2-a_3 \RA_\ZZ$ is the subgroup generated by
$2b-a_1-a_2-a_3$
(see \cite[p.~116]{CS}, for example).
As in the case of the SC lattice,
we assume  that $a_1=(a,0,0)$, $a_2=(0,a,0)$, $a_3=(0,0,a)$
in the euclidean $3$-space $\EE^3$ for
a positive real number $a$ as shown in Figure~\ref{fig:BCC01} (a).
The generator $b$ corresponds to the center point of the cube generated by
$a_1$, $a_2$ and $a_3$.

\begin{figure}[t]
 \begin{center}
\includegraphics[width=8cm]{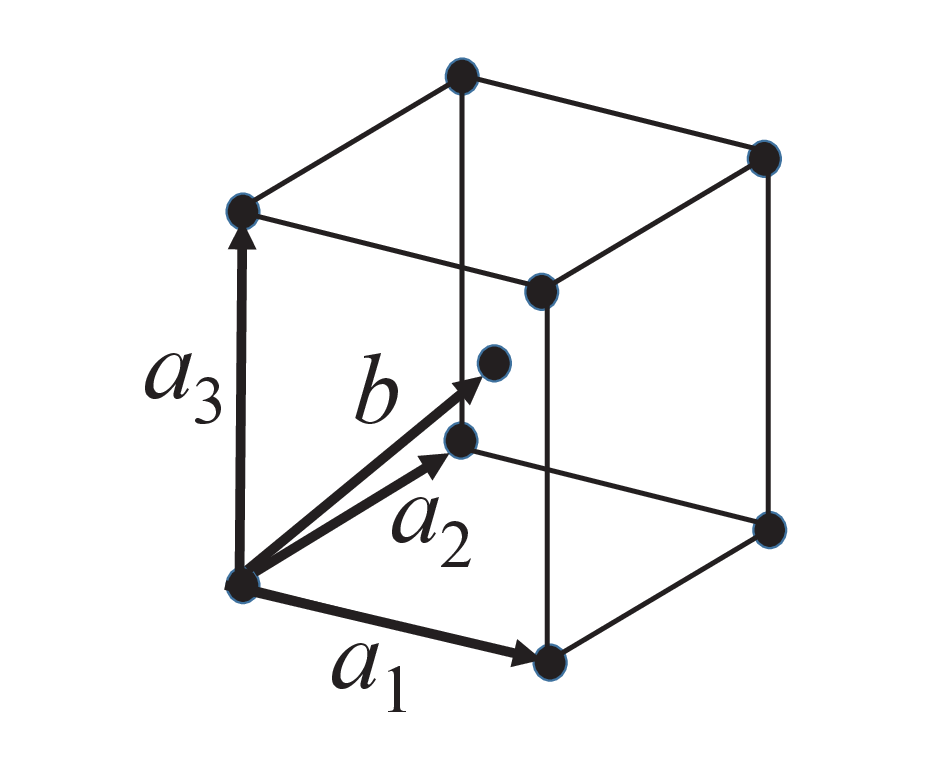}
 \end{center}
\caption{
Body centered cubic (BCC) lattice}
\label{fig:BCC01}
\end{figure}

The lattice $\AABCCd^a$ is group-isomorphic to the multiplicative group
\revsSSS{
$$
\AABCCd
:= \{\alpha_1^{n_1}\alpha_2^{n_2}\alpha_3^{n_3}\beta^{n_4} \,
| \, \mbox{abelian}, n_1, n_2, n_3, n_4 \in \ZZ, \,
\beta^2 \alpha_1^{-1} \alpha_2^{-1} \alpha_3^{-1} =1\}.
$$}
Let us denote by $\AA_4$ the multiplicative free abelian group
of rank $4$ generated by
$\alpha_1$, $\alpha_2$, $\alpha_3$ and $\beta$, i.e.,
\revsSSS{
$$
\AA_4 := \{\alpha_1^{n_1}\alpha_2^{n_2}\alpha_3^{n_3}\beta^{n_4} \,
| \, \mbox{abelian}, \, n_1, n_2, n_3, n_4 \in \ZZ\}.
$$}
Then,
$\AABCCd$ is also described as the quotient group
$$
\AABCCd = \AA_4/\LA \beta^2 \alpha_1^{-1} \alpha_2^{-1} \alpha_3^{-1} \RA,
$$
where $\LA \beta^2 \alpha_1^{-1} \alpha_2^{-1} \alpha_3^{-1} \RA$
is the (normal) subgroup generated by
$\beta^2 \alpha_1^{-1} \alpha_2^{-1} \alpha_3^{-1}$.
We shall consider the group ring
$\CC[\AABCCd]$ of $\AABCCd$,
$$
\RBCCT:=
\CC[\AABCCd] = \CC[
\alpha_1,\alpha_2, \alpha_3,
\alpha_1^{-1},\alpha_2^{-1}, \alpha_3^{-1},
\beta, \beta^{-1}]/
(\beta^2 - \alpha_1 \alpha_2 \alpha_3).
$$

\subsection{Fibering Structure of BCC Lattice}

In this subsection, we consider a projection of the BCC lattice
and its associated fibering structure.

As in the case of $\RSCT$ discussed
in Subsection~\ref{subsec:SC111},
let us consider
$$
\AABCCd_d:=\{
(\alpha_1\alpha_3^{-1})^{\ell_1}
(\alpha_2\alpha_3^{-1})^{\ell_2} \, | \,
\ell_1, \ell_2 \in \ZZ\},
$$
which is a subgroup of $\AABCCd$.
Then, we have the following decomposition as
a $\CC[\AABCCd_d]$-module.

\begin{lemma}\label{lm:BCC}
We have a natural isomorphism as
$\CC[\AABCCd_d]$-modules:
\begin{eqnarray*}
\RBCCT/(\alpha_1\alpha_2\alpha_3-1)
& \cong &
\CC[\AABCCd_d]
\oplus\CC[\AABCCd_d]\alpha_1
\oplus\CC[\AABCCd_d]\alpha_1\alpha_2\\
& & \quad \oplus
\CC[\AABCCd_d]\beta
\oplus\CC[\AABCCd_d]\alpha_1\beta
\oplus\CC[\AABCCd_d]\alpha_1\alpha_2\beta.
\end{eqnarray*}
\end{lemma}

We can prove the above lemma by using an argument
similar to that in the proof of Lemma~\ref{lm:SC111}.

\insS{
\begin{remark} \label{rmk:Snake2}
{\rm{
As in Remark~\ref{rmk:Snake1}, we can also prove
Lemma~\ref{lm:BCC} using the snake lemma \cite{L} as follows.
First, note that we have
\begin{eqnarray*}
& & \RBCCT/(\alpha_1\alpha_2\alpha_3-1) \\
& \cong & \CC[
\alpha_1,\alpha_2, \alpha_3,
\alpha_1^{-1},\alpha_2^{-1}, \alpha_3^{-1},
\beta, \beta^{-1}]/
(\beta^2 - \alpha_1 \alpha_2 \alpha_3, \,
\alpha_1 \alpha_2 \alpha_3 - 1) \\
& \cong & \CC[
\alpha_1,\alpha_2, \alpha_3,
\alpha_1^{-1},\alpha_2^{-1}, \alpha_3^{-1},
\beta, \beta^{-1}]/
(\beta^2 - 1, \, \alpha_1 \alpha_2 \alpha_3 - 1) \\
& \cong &
\CC[\alpha_1,\alpha_2, \alpha_3,
\alpha_1^{-1},\alpha_2^{-1}, \alpha_3^{-1}]/
(\alpha_1 \alpha_2 \alpha_3 - 1) \oplus
\CC[\beta, \beta^{-1}]/(\beta^2 - 1).
\end{eqnarray*}
Using the additive version
$\AABCCd^a_d$ of the abelian multiplicative group $\AABCCd_d$,
we have the commutative diagram with exact rows
$$
\xymatrix{
& 0  \ar[r]\ar[d] &\ZZ^2 \ar[r]^
-{\mbox{\tiny{$\left(
\begin{array}{ll}3&0\\0&2\end{array}
\right)$}}}
\ar_
{\mbox{\tiny{$\left(
\begin{array}{ll}1&0\\1&0\\1&0\\0&2\end{array}
\right)$}}}[d]
&3\ZZ\oplus 2\ZZ \ar[r]\ar[d] &0\\
0 \ar[r] & \AABCCd^a_d \ar[r] &\ZZ^4 \ar[r]^
{\mbox{\tiny{$\left(
\begin{array}{llll}1&1&1&0\\0&0&0&1\end{array}
\right)$}}}
&\ZZ^2 \ar[r] &0,
}
$$
where the rightmost vertical map and
the second horizontal map in the lower sequence
are the natural inclusions.
By applying the snake lemma to this diagram, we obtain the relation
in Lemma~\ref{lm:BCC} as in Remark~\ref{rmk:Snake1}.
}}
\end{remark}
}

Thus, we have the following.

\begin{proposition}\label{prop:BCC}
For
$$
\RBCCd:=\RBCCT/(\beta-1),
$$
we have a natural isomorphism as
$\CC[\AABCCd_d]$-modules:
$$
\RBCCd \cong
\CC[\AABCCd_d]
\oplus\CC[\AABCCd_d]\alpha_1
\oplus\CC[\AABCCd_d]\alpha_1\alpha_2.
$$
\end{proposition}

\begin{remark}
{\rm{
This corresponds to the triangle diagram for the projection of
the BCC lattice, which is, in fact, the same as that
depicted in Figure~\ref{fig:SC111} (c).
For the screw dislocation in the BCC lattice, $b \in \AABCCd$ coincides with
the associated Burgers vector \cite{N}.
In other words,
geometrically the ``base space''
corresponds to three sheets if we consider them
as embedded in $\EE^3$ as shown in Figure~\ref{fig:SC111}.
This comes from the discrete nature of the lattice.
However, the interval between the sheets is now given by $\sqrt{3}a/6$,
which differs from that in the SC lattice case.
More precisely, let us denote by $\fLBCCd$ the subset of $\EE^3$
corresponding to
the three sheets. Then,
we have
$$
       \fLBCCd:= \fLBCCd^{(0)} \coprod
\fLBCCd^{(1)}\coprod \fLBCCd^{(2)},
$$
where
$$
\begin{array}{rl}
       \fLBCCd^{(0)}&:=\{ \ell_1(a_1-a_3) +\ell_2(a_2-a_3) \,
                 | \, \ell_1, \ell_2 \in \ZZ\},\\
       \fLBCCd^{(1)}&:=\{ \ell_1(a_1-a_3) +\ell_2(a_2-a_3) + a_1 - b \,
                 | \, \ell_1, \ell_2 \in \ZZ\},\\
       \fLBCCd^{(2)}&:=\{ \ell_1(a_1-a_3) +\ell_2(a_2-a_3) + a_1 + a_2- b \,
                 | \, \ell_1, \ell_2 \in \ZZ\}.\\
\end{array}
$$
By observing that the degree of $b$ is equal to $3/2$,
we see that the three sheets above
correspond to the classification by the degrees (modulo $3/2$).
}}
\end{remark}

Thus, we have the following.

\begin{lemma}
As a set, $\AABCCd^a$ is also decomposed as
$$
\AABCCd^a = \AABCCd^{(0)} \coprod \AABCCd^{(1)} \coprod \AABCCd^{(2)} ,
$$
where
$$
\begin{array}{rl}
      \AABCCd^{(0)}&:=\{ \ell_1(a_1-a_3) +\ell_2(a_2-a_3) +\ell_3 b \,
                 | \, \ell_1, \ell_2, \ell_3 \in \ZZ\}, \\
      \AABCCd^{(1)}&:=\{ \ell_1(a_1-a_3) +\ell_2(a_2-a_3) + a_1 +\ell_3 b \,
                 | \, \ell_1, \ell_2, \ell_3 \in \ZZ\},\\
      \AABCCd^{(2)}&:=\{ \ell_1(a_1-a_3) +\ell_2(a_2-a_3)+ a_1+ a_2+\ell_3 b \,
                 | \, \ell_1, \ell_2, \ell_3 \in \ZZ\}.\\
\end{array}
$$
\end{lemma}

Let $\eta : \RR^3 \to \RR^3$ be the orthogonal
transformation that sends $a_1$, $a_2$ and $a_3$
to the vectors
$$a \! \left(\begin{array}{c}
\sqrt{2}/2 \\ -\sqrt{6}/6 \\
\sqrt{3}/3
\end{array}\right),
\,
a \! \left(\begin{array}{c}0 \\ \sqrt{6}/3 \\
\sqrt{3}/3
\end{array}\right)
\, \mbox{\rm and }\,
a \! \left(\begin{array}{c}
-\sqrt{2}/2 \\ -\sqrt{6}/6 \\
\sqrt{3}/3
\end{array}\right),
$$
respectively.
For a vector $\delta = (\delta_1, \delta_2, \delta_3)
\in \EE^3$,
we consider the embedding
$$
\iota^\BCC_{\delta}:
\AABCCd^a \hookrightarrow \EE^3
$$
defined by $x \mapsto \eta(x) + \delta$
for $x \in \AABCCd^a \subset \RR^3$.
Note that $\eta$ sends the vector $2b = a(1, 1, 1)$
to $\sqrt{3}a(0, 0, 1)$, and therefore $\iota^\BCC_{\delta}$
sends each $\AABCCd^{(c)}$
into a plane parallel to the plane spanned by $(1, 0, 0)$ and
$(0, 1, 0)$ for $c = 0, 1, 2$.

Let $\pi : \EE^3 \to \EE^2$ be the natural
projection to the first and second coordinates.
As in the case of $\fLSCd$,
we consider the natural embedding
$\iota^{\BCC,c}_{\bar \delta} =
\pi \circ \iota^\BCC_{\delta}|\fLBCCd^{(c)}:
\fLBCCd^{(c)} =\ZZ^2 \hookrightarrow \EE^2$, \revsSSS{
$c = 0, 1, 2$,
where $\bar \delta = (\delta_1, \delta_2) \in \EE^2$.}

With the help of Figure~\ref{fig:BCC01-2}, we can prove the
following lemma.

\begin{figure}[t]
 \begin{center}
\includegraphics[width=11cm]{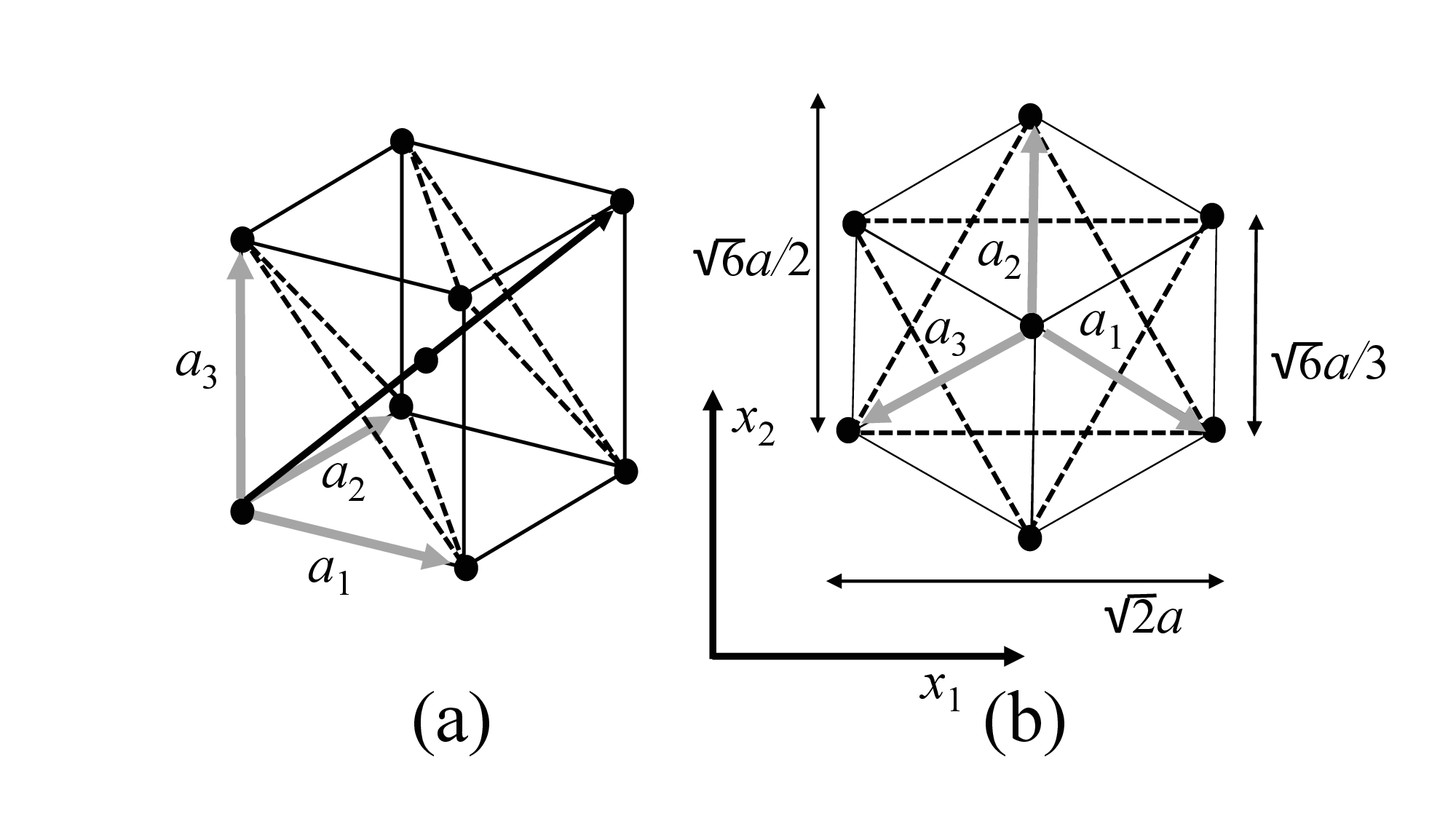}
 \newline
 \end{center}
\caption{\insOSSS{
Body centered cubic lattice and its projection along the $\revs{(1,1,1)}$-direction}}
\label{fig:BCC01-2}
\end{figure}

\begin{lemma}
For the embedding $\iota^{\BCC,c}_{\bar \delta} : \fLBCCd^{(c)}
\hookrightarrow \EE^2$, $c = 0, 1, 2$, we have the
following: \revsSSS{
\begin{eqnarray*}
\iota^{\BCC,0}_{\bar \delta}(x)
& = & (\sqrt{2}\ell_1 a+\sqrt{2}\ell_2 a/2, \sqrt{6}\ell_2 a/2) + \bar \delta \\
& & \qquad \mbox{\rm for } x = \ell_1(a_1-a_3) + \ell_2(a_2-a_3) \in \fLBCCd^{(0)}, \\
\iota^{\BCC,1}_{\bar \delta}(x)
& = & (\sqrt{2}\ell_1 a+\sqrt{2}a/2+\sqrt{2}\ell_2 a/2,
          (\sqrt{6}\ell_2 a-\sqrt{6}a/3)/2) + \bar \delta \\
& & \qquad \mbox{\rm for } x = \ell_1(a_1-a_3) + \ell_2(a_2-a_3) + a_1-b \in \fLBCCd^{(1)}, \\
\iota^{\BCC,2}_{\bar \delta}(x)
& = & (\sqrt{2}\ell_1 a+ \sqrt{2}a/2 +\sqrt{2}\ell_2 a/2,
          (\sqrt{6}\ell_2 a+2\sqrt{6}a/3)/2) + \bar \delta \\
& & \qquad \mbox{\rm for } x = \ell_1(a_1-a_3) + \ell_2(a_2-a_3) + a_1+a_2-b
\in \fLBCCd^{(2)}.
\end{eqnarray*}}
\end{lemma}

Let $\cS=\cS_+\coprod\cS_-$ be a finite subset  in $\EE^2$
as in Subsection~\ref{subsec:multi}.
In the following, we assume that $\bar\delta \in \EE^2$ satisfies
$\iota^{\BCC,c}_{\bar \delta}(\fLBCCd^{(c)})
\cap \cS= \emptyset$, $c = 0, 1, 2$.
Then, we have the following, whose proof is straightforward.

\begin{lemma}
For every fiber bundle $F_{\EE^2 \setminus \cS}$ over
$\EE^2 \setminus \cS$, by
the embedding
$\iota^{\BCC,c}_{\bar \delta}:\fLBCCd^{(c)}
\hookrightarrow \EE^2\setminus \cS$, $c=0,1,2$,
we have the pullback bundle $F_{\fLBCCd^{(c)}}$
that satisfies the commutative diagram
$$
\xymatrix{
 F_{\fLBCCd^{(c)}} \ar[d]
 \ar[r]^{\hat\iota^{\BCC,c}_{\bar \delta}}
 & F_{\EE^2\setminus \cS} \ar[d] \\
 \fLBCCd^{(c)}
\ar[r]^{\iota^{\BCC,c}_{\bar \delta}} & \EE^2\setminus \cS, \\
}
$$
where the vertical maps are the projections
of the fiber bundles and $\hat\iota^{\BCC,c}_{\bar \delta}$
is the bundle map induced by $\iota^{\BCC,c}_{\bar \delta}$.
\end{lemma}

Using this pullback diagram (Cartesian square), we have the following, whose
proof is straightforward.

\begin{lemma}
We have the following commutative diagram:
$$
\xymatrix{\EE_{\fLBCCd^{(c)}} \ar[r]^{\hat\iota^{\BCC, c}_{\bar\delta}}
\ar[d]
\ar@/_24pt/[dd]_{\widehat{\psi}}
& \EE_{\EE^2\setminus \cS}
\ar[d]
\ar@/^24pt/[dd]^{\widehat{\psi}} \\
\fLBCCd^{(c)} \ar[r]^{\iota^{\BCC, c}_{\bar\delta}} &
\EE^2 \setminus \cS \\
S^1_{\fLBCCd^{(c)}} \ar[u] \ar[r]^{\hat\iota^{\BCC, c}_{\bar\delta}} &
S^1_{\EE^2 \setminus \cS}
\ar[u]
}
$$
for $c = 0, 1, 2$, where 
\revsSSS{$\EE_{\fLBCCd^{(c)}}$ and $S^1_{\fLBCCd^{(c)}}$
are the trivial bundles},
the straight vertical arrows are projections
of the fiber bundles, $\widehat{\psi}$ are the bundle
maps induced by $\psi$ defined in Section~\textup{\ref{section2}},
and the positive constant $d$ in
Section~\textup{\ref{section2}} is now
set $d=\sqrt{3}a/2$, which $\widehat{\psi}$ depends on.
\end{lemma}

%

The following proposition corresponds to the case where $\cS=\emptyset$
and the proof is left to the reader.

\begin{proposition}
Set $\gamma=\exp(4\pi\ii \delta_3/(\sqrt{3}a)) \in S^1$
and consider the global sections \revsSSSS{
$$\tilde \sigma_{\gamma, c}\in
\Gamma(\fLBCCd^{(c)}, S^1_{\fLBCCd^{(c)}}), \quad
c = 0, 1, 2,$$}
that constantly take the values $\gamma \zeta_3^{-c}$,
where $\zeta_3=\exp(2\pi\ii /3)$.
Then, we have \revsSSSS{
$$
\iota^{\BCC}_{\delta}(\AABCCd^a)
       =\bigcup_{c=0}^2
\hat\iota^{\BCC,c}_{\bar \delta}
\left(\widehat\psi^{-1}
\left(\tilde \sigma_{\gamma,c}(\fLBCCd^{(c)})\right)\right)
\subset \EE^3.
$$}
\end{proposition}

\subsection{Algebraic Description of Screw Dislocations in BCC Lattice}\label{subsection4.4}

Recall that a screw dislocation in the BCC lattice is basically given by
the $\revs{(1,1,1)}$-direction.
In other words, the Burgers vector is parallel to the
$\revs{(1,1,1)}$-direction, or
more precisely it coincides with $b$ itself,
\insMMM{
up to automorphisms of the BCC lattice} \cite{N}.

In the following, for
$\gamma' \in S^1$ and \revsSSSS{$\tilde \sigma \in
\Gamma(\fLBCCd, S^1_{\fLBCCd})$ expressed as
$\tilde \sigma(x) = (s(x), x)$ for $x \in \fLBCCd$,
we define their multiplication $\gamma' \tilde \sigma \in
\Gamma(\fLBCCd, S^1_{\fLBCCd})$ by
$(\gamma' \tilde \sigma)(x) = (\gamma' s(x), x)$},
$x \in \fLBCCd$, \revsSSS{where
$S^1_{\fLBCCd}$ is the trivial $S^1$-bundle over $\fLBCCd$.}

Now, as in Proposition~\ref{prop:SC-single} for the SC lattice case,
we have the following description of a single
screw dislocation in the BCC lattice. \revsSSS{
In the following, we set $z_0' = z_0 - (\delta_1 + \delta_2 \ii)$.}

\revsSSS{
\begin{proposition}\label{prop:singleBCC}
The single screw dislocation expressed by \revsSSSS{
$$
\bigcup_{c=0}^2
\hat\iota^{\BCC,c}_{\bar \delta} \left(
\widehat\psi^{-1}
\left((\gamma\zeta_3^{-c} \tilde \sigma_{z_0'})(\fLBCCd^{(c)})\right)\right)
$$
around $z_0 \in \EE^2$ is a subset of $\EE^3$,
where
$\tilde \sigma_{z_0'}$ is an element of
$\Gamma(\fLBCCd, S^1_{\fLBCCd})$}
given by
\revsSSSS{
$$
\begin{array}{l}
\tilde \sigma_{z_0'}(x) = \\
\displaystyle{
 \left(
 \frac{\sqrt{2}\ell_1 a+\sqrt{2}\ell_2 a/2-x_0' +(\sqrt{6}\ell_2 a/2 - y_0')\ii}
        {|\sqrt{2}\ell_1 a+\sqrt{2}\ell_2 a/2-x_0' +(\sqrt{6}\ell_2 a/2 - y_0')\ii|},
       x \right)}
\\
\qquad \mbox{\rm for } x = \ell_1(a_1-a_3) + \ell_2(a_2-a_3) \in \fLBCCd^{(0)}, \\
\  \\
\tilde \sigma_{z_0'}(x) = \\
\displaystyle{
 \left(
 \frac{\sqrt{2}\ell_1 a+\sqrt{2}a/2+\sqrt{2}\ell_2 a/2-x_0'
          +((\sqrt{6}\ell_2 a-\sqrt{6}a/3)/2 - y_0')\ii}
             {|\sqrt{2}\ell_1 a+\sqrt{2}a/2+\sqrt{2}\ell_2 a/2-x_0'
          +((\sqrt{6}\ell_2 a-\sqrt{6}a/3)/2 - y_0')\ii|}, x \right)}
\\
\quad \mbox{\rm for } x = \ell_1(a_1-a_3) + \ell_2(a_2-a_3)
+ a_1-b \in \fLBCCd^{(1)}, \\
\ \\
\tilde \sigma_{z_0'}(x) = \\
\displaystyle{
 \left(
 \frac{\sqrt{2}\ell_1 a+\sqrt{2}a/2+\sqrt{2}\ell_2 a/2-x_0'(
          +((\sqrt{6}\ell_2 a+2\sqrt{6}a/3)/2 - y_0')\ii}
             {|\sqrt{2}\ell_1 a+ \sqrt{2}a/2+\sqrt{2}\ell_2 a/2-x_0'
          +((\sqrt{6}\ell_2 a+2\sqrt{6}a/3)/2 - y_0')\ii|}, x \right)}
 \\
\quad \mbox{\rm for } x = \ell_1(a_1-a_3) + \ell_2(a_2-a_3) + a_1+a_2-b
\in \fLBCCd^{(2)},
\end{array}
$$
}
where $z_0'=x_0'+y_0'\ii \in \CC$
and $(\ell_1, \ell_2) \in \ZZ^2$.
\end{proposition}}

Furthermore,
for \revsSSSS{
$\tilde \sigma_1$ and $\tilde \sigma_2 \in
\Gamma(\fLBCCd, S^1_{\fLBCCd})$ expressed as
$\tilde \sigma_a(x) = (s_a(x), x)$ for $x \in \fLBCCd$,
$a =1,2$,
we define their multiplication $\tilde \sigma_1 \tilde \sigma_2$ by
$(\tilde \sigma_1 \tilde \sigma_2)(x) = (s_1(x)s_2(x), x)$,
$x \in \fLBCCd$.}
Using the multiplication,
we have the following description of a parallel
multi-screw dislocation in the BCC lattice.

\revsSSSS{
\begin{proposition} \label{prop:multiBCC}
The parallel multi-screw dislocation in the BCC lattice given by \insOSSS{
\begin{eqnarray*}
& &
\bigcup_{c=0}^2 \hat{\iota}^{\BCC,c}_{\bar \delta}\left(
\widehat{\psi}^{-1}
\left(\left(\gamma\zeta_3^{-c} \prod_{z_i'\in \cS_+'} \tilde \sigma_{z_i'}
\prod_{z_j'\in \cS_-'} \overline{\tilde \sigma_{z_j'}}
\right)(\fLBCCd^{(c)})\right)\right) \\
& = &
\bigcup_{c=0}^2
\hat\iota^{\BCC,c}_{\bar \delta} \left(
\frac{\sqrt{3}a}{4\pi \ii}\exp^{-1}
\left(\left(\gamma\zeta_3^{-c}
\prod_{z_i'\in \cS_+'} \tilde \sigma_{z_i'}
\prod_{z_j'\in \cS_-'} \overline{\tilde \sigma_{z_j'}}\right)
(\fLBCCd^{(c)})
\right)
\right)
\end{eqnarray*}}
is a subset of $\EE^3$,
where $\cS$ corresponds to
the position of the dislocation lines,
$z_i' = z_i - (\delta_1 + \delta_2 \ii)$,
$z_j' = z_j - (\delta_1 + \delta_2 \ii)$, $\cS_+'
= \{z_i' \,|\, z_i \in \cS_+\}$, $\cS_-' = \{z_j' \,|\,
z_j \in \cS_-\}$ and $\cS' = \cS_+' \coprod \cS_-'$.
\end{proposition}}

\insOSSS{In Subsection~\ref{subsection621}, we will give a more
explicit formula for a single screw dislocation.}

\section{Energy of Screw Dislocation}\label{section5}

\subsection{Energy of Screw Dislocation in SC Lattice}\label{subsection5.1}

In this section, we consider the
\ins{strain} energy of a single screw dislocation in the
SC lattice along the $\revs{(0,0,1)}$-direction as discussed
in Subsection~\ref{subsec:sdSC}.
We adopt a spring model, in which certain ``edges'' of the SC lattice correspond
to elastic springs, whose natural lengths are equal to $a$ or $\sqrt{2}a$.

More precisely, in our model, we have the elastic springs on the edges \revsSSS{
$$
[(n_1, n_2,n_3),(n_1+1, n_2,n_3)], \quad
[(n_1, n_2,n_3),(n_1, n_2+1,n_3)],
$$
$$
[(n_1, n_2,n_3),(n_1, n_2,n_3+1)], \quad
[(n_1, n_2,n_3),(n_1+1, n_2,n_3 \pm 1)],
$$
$$
[(n_1, n_2,n_3),(n_1, n_2+1,n_3\pm 1)],\quad
[(n_1, n_2,n_3),(n_1+1, n_2\pm 1,n_3)],
$$
for all $(n_1,n_2,n_3) \in \ZZ^3$.}
\insS{Note that the above parametrization refers to
a local one for the SC lattice after the dislocation, in a neighborhood
of each vertex sufficiently far from the dislocation line.
It is clear that such a parametrization
does not work globally; however, around each point,
such a parametrization works as long as the point
is far from the dislocation center.
In this section we will
use this parameterization for simplicity
and compute the strain energy caused by a screw dislocation.}

We note that there are several other possibilities
for the choice of the edges.
The choice of the model, however, does not affect the
essentials of the results,
as we will see later.

Now, let us consider the screw dislocation
in the SC lattice along the $\revs{(0,0,1)}$-direction around $z_0 \in \CC$,
as described in Subsection~\ref{subsec:sdSC}.
\revsSSS{For simplicity, throughout 
this section, we suppose $\delta = (0, 0, 0)$ and
$\bar \delta = (0, 0)$ so that $z_0' = z_0$. 
Furthermore, we assume that $z_0 = x_0 + y_0 \ii$,
$x_0, y_0 \in \RR$, satisfies that
$x_0, y_0$ and $\pm x_0 \pm y_0$ are not
integral multiples of $a$.}
We regard that the
original lattice is in a stable position, and we consider the
\ins{elastic} energy resulting from the dislocation. For this,
we need to investigate the difference between
the original position and the position resulting from the
dislocation.

\revsSSS{In the following, we also assume that $\gamma = 1$ for simplicity.
For notational convention, we will denote $\check\sigma_{z_0,\gamma}$
etc.\ simply by $\check\sigma_{z_0}$ etc.\ 
by suppressing $\gamma$.}

First, we define
the relative height \revsSSS{differences
$\epsilon_{n_1,n_2}^{(1)}$,
$\epsilon_{n_1,n_2}^{(2)}$ and
$\epsilon_{n_1,n_2}^{(\pm)}$
by 
\begin{equation}\label{eq:dfn_of_eps}
  \begin{array}{rl}
    \displaystyle{\epsilon_{n_1,n_2}^{(1)}} & \displaystyle{= \frac{a}{2\pi\ii}
    \left(\log(\check\sigma_{z_0}((n_1+1)a, n_2 a))
    -\log(\check\sigma_{z_0}(n_1 a, n_2 a)) \right), }
    \raisebox{0mm}[7mm][7mm]{} \\
    \displaystyle{\epsilon_{n_1,n_2}^{(2)}} & \displaystyle{= \frac{a}{2\pi\ii}
    \left(\log(\check\sigma_{z_0}(n_1 a, (n_2+1)a))
    -\log(\check\sigma_{z_0}(n_1 a, n_2 a)) \right), }
    \raisebox{0mm}[7mm][7mm]{} \\
     \displaystyle{\epsilon_{n_1,n_2}^{(\pm)}} & \displaystyle{= \frac{a}{2\pi\ii}
    \left(\log(\check\sigma_{z_0}((n_1+1)a, (n_2\pm 1)a))
    -\log(\check\sigma_{z_0}(n_1 a, n_2 a)) \right),} \\
  \end{array}
\end{equation}
respectively,
where $\log x = \log_e x$ for $x \in S^1$ is considered to be
$\ii$ times the argument of $x$, and we choose the values so that
$-a/2 < \epsilon_{n_1,n_2}^{(i)} \leq a/2$
for $i = 1, 2$ and $\pm$. Later we will see that
it never takes the value $a/2$.}

In what follows, for
a section $\check \sigma \in \Gamma(\fLSCp, S^1_{\fLSCp})$
expressed as $\check \sigma(x) =
(s(x), x)$ for $x \in \fLSCp$, we often
use the symbol $\check \sigma(x)$ instead of $s(x)$
by abuse of notation.
We recall that \revsSSS{
$$
\check\sigma_{z_0}(n_1 a,n_2 a)
=\revs{\check\sigma_{z_0,\gamma}(n_1 a,n_2 a)}
  = \frac{n_1 a - x_0+ (n_2 a - y_0)\ii}{\abs{n_1 a - x_0+ (n_2 a - y_0)\ii}}
$$
for $z_0 =x_0 + y_0 \ii$, $x_0, y_0 \in \RR$,
and $\gamma=1$ in Proposition~\ref{prop:SC-single}.
Set $z = n_1 a - x_0+ (n_2 a - y_0) \ii$, which
is not a real multiple of $1, \ii$ or $1 \pm \ii$ by
our assumption on $z_0$. Then, 
we get
\begin{equation}\label{eq:a/z}
  \begin{array}{rl}
    \displaystyle{\epsilon_{n_1,n_2}^{(1)}} & \displaystyle{= \frac{a}{2\pi\ii}
    \log \frac{1+a/z}{|1+a/z|}, }
    \raisebox{0mm}[7mm][7mm]{} \\
    \displaystyle{\epsilon_{n_1,n_2}^{(2)}} & \displaystyle{= \frac{a}{2\pi\ii}
     \log \frac{1+a\ii/z}{|1+a\ii/z|}, }
        \raisebox{0mm}[7mm][7mm]{} \\
     \displaystyle{\epsilon_{n_1,n_2}^{(\pm)}} & \displaystyle{= \frac{a}{2\pi\ii}
     \log \frac{1+a(1 \pm \ii)/z}{|1+a(1 \pm \ii)/z|}}. \\
  \end{array}
\end{equation}
Note that $1+a/z$, $1+a\ii/z$ and
$1+a(1 \pm \ii)/z$ have non-zero imaginary parts,
and therefore
$\epsilon_{n_1,n_2}^{(i)}$ never takes the value $a/2$
for $i = 1, 2$ and $\pm$.
}

Then, the difference of length in each segment \revsSSS{
$$[(n_1, n_2,n_3),(n_1+1, n_2,n_3)] \quad \mbox{\rm or} \quad
[(n_1, n_2,n_3),(n_1, n_2+1,n_3)]
$$}
is given by \revsSSS{
$$
\Delta_{n_1, n_2}^{(i)} =\sqrt{a^2 +
(\epsilon_{n_1,n_2}^{(i)})^2}-a,
$$}
$i=1,2$, whereas
the difference of length in each diagonal segment \revsSSS{
$$[(n_1, n_2,n_3),(n_1+1, n_2, n_3\pm 1)]
\quad \mbox{\rm or} \quad
[(n_1, n_2, n_3),(n_1, n_2+1, n_3\pm 1)]$$}
is given by
\revsSSS{
\begin{equation}
\Delta_{n_1, n_2}^{d(i, \pm)}
=\sqrt{(a \pm \epsilon_{n_1,n_2}^{(i)})^2
+a^2}-\sqrt{2}a,
\label{eq:De1}
\end{equation}
}
$i = 1, 2$,
and the difference of length in the diagonal segment \revsSSS{
$$[(n_1, n_2,n_3),(n_1+1, n_2\pm1,n_3)]$$}
is given by
\revsSSS{
\begin{equation}
\Delta_{n_1, n_2}^{d(\pm)} =
\sqrt{2a^2 + (\epsilon_{n_1,n_2}^{(\pm)})^2
}-\sqrt{2}a.
\label{eq:De2}
\end{equation}
}
On the other hand,
the length of the segment \revsSSS{
$[(n_1, n_2,n_3),(n_1, n_2,n_3+1)]$}
is constantly equal to
the natural length $a$ and thus
we set \revsSSS{$\Delta_{n_1, n_2}^{(3)} =0$.}

Then, we have the following.

\begin{lemma} \label{lm:elasap}
If \revsSSS{
$$\frac{a}{\sqrt{(n_1 a-x_0)^2 + (n_2 a-y_0)^2}}$$}
is sufficiently small, then \revsSSS{
$\epsilon^{(1)}_{n_1,n_2}$,
 $\epsilon^{(2)}_{n_1,n_2}$
and $\epsilon^{(\pm)}_{n_1,n_2}$}
are approximately given by
\revsSSS{
\begin{eqnarray}
    \epsilon_{n_1,n_2}^{(1)} & = & - \frac{a}{2\pi}
    \frac{a(n_2 a-y_0)}{(n_1 a-x_0)^2
+ (n_2 a - y_0)^2} \nonumber \\
& & \qquad \qquad \qquad + o\left(
\frac{a}{\sqrt{(n_1 a -x_0)^2+(n_2 a -y_0)^2}}
\right), \nonumber \\
    \epsilon_{n_1,n_2}^{(2)} & = & -\frac{a}{2\pi}
    \frac{a(n_1 a -x_0)}{(n_1 a -x_0)^2
    + (n_2 a - y_0)^2}  \nonumber \\
   & & \qquad \qquad \qquad
+ o\left( \frac{a}{\sqrt{(n_1 a -x_0)^2+(n_2 a -y_0)^2}}
\right), \label{eq:approx_of_eps} \\
    \epsilon_{n_1,n_2}^{(\pm)} & = & -\frac{a}{2\pi}
    \frac{\pm a(n_1 a -x_0)
    + a(n_2 a -y_0)}
          {(n_1 a -x_0)^2 + (n_2 a - y_0)^2} \nonumber \\
   & & \qquad \qquad \qquad
+ o\left( \frac{a}{\sqrt{(n_1 a -x_0)^2
 +(n_2 a -y_0)^2}}
    \right), \nonumber
\end{eqnarray}}
respectively,
whereas \revsSSS{$\Delta_{n_1, n_2}^{(i)}$
$\Delta_{n_1, n_2}^{d(i, \pm)}$ and
$\Delta_{n_1, n_2}^{d(\pm)}$}
are approximately given by \revsSSS{
\begin{equation}\label{eq:approx_of_delta}
  \begin{array}{l}
    \displaystyle{\Delta_{n_1, n_2}^{(i)} = \frac{1}{2a}
    (\epsilon_{n_1,n_2}^{(i)})^2
 + o\left( \frac{a}{\sqrt{(n_1 a -x_0)^2+(n_2 a -y_0)^2}} \right),}
     \raisebox{0mm}[7mm][7mm]{} \\
     \displaystyle{\Delta_{n_1, n_2}^{d(i, \pm)}
= \pm \frac{1}{\sqrt{2}} %
    \epsilon_{n_1,n_2}^{(i)}
+ o\left( \frac{a}{\sqrt{(n_1 a -x_0)^2+(n_2 a -y_0)^2}} \right),}
    \raisebox{0mm}[7mm][7mm]{} \\
    \displaystyle{\Delta_{n_1, n_2}^{d(\pm)} = \frac{1}{2\sqrt{2}a} %
    (\epsilon_{n_1,n_2}^{(\pm)})^2
+ o\left( \frac{a}{\sqrt{(n_1 a -x_0)^2+(n_2 a -y_0)^2}} \right),}
  \end{array}
\end{equation}}
respectively, $i = 1, 2$.
\end{lemma}

\begin{proof}
\revsSSS{
Setting
  \[
    z := (n_1 a -x_0)+(n_2 a -y_0) \ii,
  \]
we have
\begin{eqnarray*}
    \check \sigma_{z_0}((n_1+1)a, n_2 a)
    & = & \dfrac{z+a}{\abs{z+a}}
    = \check\sigma_{z_0}(n_1 a,n_2 a) \dfrac{1+a/z}{\abs{1+a/z}}, \\
    \check\sigma_{z_0}(n_1 a, (n_2+1)a)
    & = & \dfrac{z+a \ii}{\abs{z+a \ii}}
    = \check\sigma_{z_0}(n_1 a,n_2 a)
          \dfrac{1+a \ii/z}{\abs{1+a \ii/z}},
    \\
    \check\sigma_{z_0}((n_1+1)a, (n_2 \pm 1)a)
    & = & \dfrac{z+a\pm a \ii}{\abs{z+a\pm a \ii}} \\
    & = & \check\sigma_{z_0}(n_1 a,n_2 a)
    \dfrac{1+a(1\pm\ii)/z}{\abs{1+a(1\pm\ii)/z}}.
\end{eqnarray*}
}By Taylor expansion, we have, for $w = \xi + \xi' \ii$,
$\xi, \xi' \in \RR$,
$$
    \arg (1+w) = \arctan \dfrac{\xi'}{1+\xi} = \xi' + o(\abs{w})
$$
as $w \to 0$. Therefore, we obtain that
\revsSSS{
$$
\begin{array}{rl}
      \arg \left( 1+\dfrac{a}{z} \right) %
      &= \Im \dfrac{a}{z} + o \left( \dfrac{a}{\abs{z}} \right) \\
      &= -\dfrac{a}{\abs{z}^2}(n_2 a -y_0) + o \left( \dfrac{a}{\abs{z}} \right), \\
      \arg \left( 1+ \dfrac{a \ii}{z} \right) %
      &= \Re \dfrac{a}{z} + o \left( \dfrac{a}{\abs{z}} \right) \\
      &= \dfrac{a}{\abs{z}^2}(n_1 a -x_0) + o \left( \dfrac{a}{\abs{z}} \right),
\\
      \arg \left( 1+ \dfrac{a(1\pm\ii)}{z} \right) %
      &= \Im \dfrac{a}{z} \pm \Re \dfrac{a}{z} + o \left( \dfrac{a}{\abs{z}} \right) \\
      &= \dfrac{a}{\abs{z}^2}
   \left(-(n_2 a-y_0)\pm(n_1 a -x_0)\right)
 + o \left( \dfrac{a}{\abs{z}} \right).\\
\end{array}
$$
}
Then, we get the approximation formula (\ref{eq:approx_of_eps}) from
definition (\ref{eq:dfn_of_eps}).

Finally, we can prove the approximation formula (\ref{eq:approx_of_delta})
for \revsSSS{$\Delta_{n_1, n_2}^{(i)}$, $\Delta_{n_1,n_2}^{d(i, \pm)}$,
$i = 1, 2$, and $\Delta_{n_1,n_2}^{d(\pm)}$}
by simple application of the Taylor expansion.
This completes the proof.
\end{proof}

Following the spirit of the theory of elasticity \cite{LL,M}, in the following,
we assume that \revsSSS{
\revsMMM{
$$
   \frac{a}{\sqrt{(n_1a-x_0)^2+(n_2a-y_0)^2}} = \frac{a}{|z|}
$$
}
is small in Lemma~\ref{lm:elasap} and in particular
$a/|z| < 1/\sqrt{2}$.}
This assumption means that the node \revsSSS{
$(n_1 a, n_2 a)$} in $\EE^2$ is sufficiently far from the
center $(x_0,y_0)$ of dislocation relative to the lattice length $a$.
Such an approximation does not hold for the nodes near the center.
More explicitly, the approximation given above is valid
for the \insSS{elastic} energy in the far region
\revsSSS{\begin{equation}
  A_{\rho} := \left\{ (n_1,n_2) \in \ZZ^2\,|\,
\rho a < \sqrt{(n_1 a -x_0)^2+(n_2 a -y_0)^2} \right\}
\label{eq:Arho}
\end{equation}}
for sufficiently large fixed \revsSSS{$\rho \geq \sqrt{2}$}. On the other hand,
in the core region $\ZZ^2 \setminus A_{\rho}$, the approximation
fails and we need to adopt another approach.

\insS{
Furthermore, for later convenience, let us introduce the notation
\revsSSS{\begin{equation}
  A_{\rho, N} := \left\{ (n_1,n_2) \in \ZZ^2\,|\,
\rho a < \sqrt{(n_1 a -x_0)^2+(n_2 a -y_0)^2}
< N a \right\}
\label{eq:ArhoN}
\end{equation}}for $N > \rho$, which is bounded
and is a finite set.
}

\insS{
We can now compute the elastic energy caused
by the screw dislocation.
Since our model has the translational symmetry along the
$(0,0,1)$-axis (i.e.,
the set of lattice points \revsSSS{$(n_1, n_2, n_3)$}
together with the edges with springs
in our model is invariant under
the translation \revsSSS{$n_3 \mapsto n_3 +1$}),
we will concentrate ourselves
on the energy density for unit length in the
$(0,0,1)$-direction, and call it simply the elastic
energy of dislocation again.
}

Let $k_p$ and $k_d$ be spring constants of the horizontal
springs and the diagonal springs, respectively.
Then, the \ins{elastic}
energy of dislocation in the 
\ins{annulus region}
$A_{\rho,N}$ is given by
\revsSSS{
\begin{equation}
E_{\rho,N}(x_0,y_0) := \sum_{(n_1, n_2)\in A_{\rho,N}}
\cE_{n_1,n_2},
\label{eq15aa}
\end{equation}}
where \revsSSS{
$\cE_{n_1,n_2}$} is the energy density
defined by \revsSSS{
\begin{eqnarray}
\cE_{n_1,n_2}& := & 
\frac{1}{2}k_p
\biggl(\left(\Delta_{n_1, n_2}^{(1)}\right)^2
+\left(\Delta_{n_1, n_2}^{(2)}\right)^2 \biggr)
\nonumber \\
& & \quad + \frac{1}{2}k_d
\biggl(
\left(\Delta_{n_1, n_2}^{d(1, +)}\right)^2
+\left(\Delta_{n_1, n_2}^{d(2, +)}\right)^2
+\left(\Delta_{n_1, n_2}^{d(1, -)}\right)^2  \nonumber \\
& & \quad
+\left(\Delta_{n_1, n_2}^{d(2, -)}\right)^2
+\left(\Delta_{n_1, n_2}^{d(+)}\right)^2
+\left(\Delta_{n_1, n_2}^{d(-)}\right)^2\biggr).
\raisebox{0mm}[4mm][4mm]{}
\nonumber
\end{eqnarray}
}

\revsSSS{
\begin{proposition} \label{prop:9}
\begin{itemize}
\item[$(1)$] For $(n_1, n_2) \in A_\rho$,
the energy density $\cE_{n_1, n_2}$ is
expressed by \revsSSSS{a real analytic function $\cE(w, \overline{w})$ of
$w$ and $\bar{w} \in \CC$} with $|w| < 1/\sqrt{2}$ in such a way that
\begin{eqnarray*}
& & \cE_{n_1, n_2} \\
& = & \cE\left(
\frac{a}{(n_1 a - x_0)+(n_2 a - y_0) \ii},
\frac{a}{(n_1 a - x_0)-(n_2 a - y_0) \ii}
\right).
\end{eqnarray*}
\item[$(2)$] Let us consider the power series expansion
$$
\cE(w,\overline w) = \sum_{s=0}^\infty
\cE^{(s)}(w, \overline w), \quad
\cE^{(s)}(w, \overline w) := \sum_{i+j=s, i,j\ge0} C_{i, j} 
w^i \overline{w}^j,
$$
for some $C_{i, j} \in \CC$.
Then, we have the following:
\begin{itemize}
\item[\textup{(a)}]  
$\cE^{(0)}(w, \overline w)=\cE^{(1)}(w, \overline w)=0$,
\item[\textup{(b)}]  The leading term is given by 
\begin{eqnarray}
&&\cE^{(2)}(w, \overline w) =
\frac{a^2}{8\pi^2} k_d 
\revsSSS{
w\overline{w},}
\nonumber \\
&&\cE^{(2)}\left(
\frac{a}{(n_1 a - x_0)+(n_2 a - y_0)\ii},
\frac{a}{(n_1 a - x_0)-(n_2 a - y_0)\ii}
\right)\nonumber \\
&&=
\frac{1}{8\pi^2} k_d
 \left[
  \frac{a^4}{(n_1 a -x_0)^2 +(n_2 a -y_0)^2}
 \right], 
\label{eq15a}
\end{eqnarray}
\item[\textup{(c)}]  $C_{i, j}=\overline{C_{j, i}}$, and
\item[\textup{(d)}]  for every $s \geq 2$, there is a constant $M_s > 0$ such that
$$
|\cE^{(s)}(w, \overline w)| \le M_s |w|^s. 
$$
\end{itemize}
\end{itemize}
\end{proposition}
}

\begin{proof}
\revsSSS{
Set
$$w = \frac{a}{(n_1 a - x_0) + (n_2 a - y_0) \ii} = \frac{a}{z}.$$
Note that $|w| < 1/\sqrt{2}$ as we have assumed $\rho \geq \sqrt{2}$.
We have seen in (\ref{eq:a/z}) that
\begin{equation}
  \begin{array}{rl}
    \displaystyle{\epsilon_{n_1,n_2}^{(1)}} & \displaystyle{= \frac{a}{2\pi\ii}
    \log \frac{1+w}{|1+w|}, }
    \raisebox{0mm}[7mm][7mm]{} \\
    \displaystyle{\epsilon_{n_1,n_2}^{(2)}} & \displaystyle{= \frac{a}{2\pi\ii}
     \log \frac{1+w\ii}{|1+w\ii|}, }
        \raisebox{0mm}[7mm][7mm]{} \\
     \displaystyle{\epsilon_{n_1,n_2}^{(\pm )}} & \displaystyle{= \frac{a}{2\pi\ii}
     \log \frac{1+w(1 \pm \ii )}{|1+w(1 \pm \ii )|}}. \\
  \end{array}
\end{equation}
Note that if we consider $w$ as a complex variable
in $\CC$ with $|w| < 1/\sqrt{2}$, then $\epsilon_{n_1,n_2}^{(i)}$
are \revsSSSS{real analytic functions of $w$ and $\bar{w}$, 
$i = 1, 2, \pm$.}
Furthermore, we have
\begin{eqnarray*}
\cE_{n_1, n_2} & = & \frac{1}{2}k_p
\left(\left(\sqrt{a^2 + (\epsilon_{n_1,n_2}^{(1)})^2} - a\right)^2 +
\left(\sqrt{a^2 + (\epsilon_{n_1,n_2}^{(2)})^2} - a\right)^2 \right) \\
& & \quad + \frac{1}{2}k_d 
\left(\left(\sqrt{(a+\epsilon_{n_1,n_2}^{(1)})^2+a^2} - \sqrt{2}a\right)^2 \right.\\
& & \quad +
\left(\sqrt{(a+\epsilon_{n_1,n_2}^{(2)})^2+a^2} - \sqrt{2}a\right)^2 \\
& & \quad + \left(\sqrt{(a-\epsilon_{n_1,n_2}^{(1)})^2+a^2} - \sqrt{2}a\right)^2 \\
& & \quad +
\left(\sqrt{(a-\epsilon_{n_1,n_2}^{(2)})^2+a^2} - \sqrt{2}a\right)^2 \\
& & \quad + \left(\sqrt{2a^2 + (\epsilon_{n_1,n_2}^{(+)})^2} - \sqrt{2}a\right)^2 \\
& & \left.\quad +
\left(\sqrt{2a^2 + (\epsilon_{n_1,n_2}^{(-)})^2} - \sqrt{2}a\right)^2
\right).
\end{eqnarray*}
Thus, $\cE_{n_1, n_2}$ can be considered to be \revsSSSS{a real
analytic function of
$w$ and $\bar w \in \CC$} with $|w| < 1/\sqrt{2}$.
}

\revsS{
(2): Items (a) and (b) are obtained by straightforward 
calculations as follows: \revsSSS{
\begin{eqnarray}
\cE_{n_1, n_2} 
& = & k_d
\left[
 \frac{1}{2} (\epsilon_{n_1,n_2}^{(1)})^2
 +\frac{1}{2} (\epsilon_{n_1,n_2}^{(2)})^2\right. \nonumber \\
& & \qquad \qquad \qquad
+ \left. o \left( \frac{a^2}{(n_1 a -x_0)^2+(n_2 a -y_0)^2} \right) \right] \nonumber \\
& = &
\frac{1}{8\pi^2} k_d
 \left[
  \frac{a^4}{(n_1 a -x_0)^2 +(n_2 a -y_0)^2}
 \right. \nonumber \\
& & \qquad \qquad \qquad
+ \left. o\left( \frac{a^2}{(n_1 a -x_0)^2+
               (n_2 a -y_0)^2} \right)\right].
\label{eq15}
\end{eqnarray}}
Since the energy density is a real number, we obtain the
relation in item (c). The analyticity in item (1) implies (d).
This completes the proof.}
\end{proof}

\bigskip

\revsSSS{
As the summation in (\ref{eq15aa}) is finite, we have
\begin{eqnarray}
E_{\rho,N}(x_0,y_0) & = &  
\sum_{s=2}^\infty
\sum_{(n_1, n_2)\in A_{\rho,N}}
\cE^{(s)}\left(
\frac{a}{(n_1 a - x_0)+(n_2 a - y_0) \ii},
\right. \nonumber \\
& & 
\left.
\qquad\qquad\qquad\qquad
\frac{a}{(n_1 a - x_0)-(n_2 a - y_0) \ii}
\right).
\label{eq-series}
\end{eqnarray}
}
\revsSSS{
At this stage,
it seems difficult to estimate the whole series: however,
we can estimate each term in this series
using
the \textit{truncated Epstein-Hurwitz zeta function} 
$\zeta_{\rho, N}(s, z_0)$ defined by
\begin{equation}
    \zeta_{\rho, N}(s, z_0) := \sum_{(n_1, n_2)\in A_{\rho, N}}
    \frac{1}{((n_1+x_0)^2 +(n_2+y_0)^2)^{s/2}},
    \label{eq:tzeta}
\end{equation}
\revsMM{
where $z_0:=x_0+y_0 \ii$.
}
In particular, we have the following theorem for the 
``principal part'' of the elastic energy.
}
\insS{
\begin{theorem}\label{thm:energy}
\revsSSS{
The principal part of the elastic energy $E_{\rho,N}(x_0,y_0)$, defined by
\begin{eqnarray*}
  E^{(\mathrm{p})}_{\rho,N}(x_0,y_0) &:=&
   \sum_{(n_1, n_2) \in A_{\rho,N}} 
\cE^{(2)}\left(
\frac{a}{(n_1 a - x_0)+(n_2 a - y_0) \ii},\right.\\
&& \qquad\qquad\qquad\qquad\left.
\frac{a}{(n_1 a - x_0)-(n_2 a - y_0)\ii}
\right)  \\ 
&=& \frac{1}{8\pi^2} k_d
   \sum_{(n_1, n_2) \in A_{\rho,N}}
   \left[
    \frac{a^4}{(n_1 a -x_0)^2 +(n_2 a -y_0)^2}
   \right]
\end{eqnarray*}
is given by
\begin{equation}
E^{(\mathrm{p})}_{\rho,N}(x_0,y_0)
= \frac{1}{8\pi^2}k_d a^2
\zeta_{\rho, N}(2, -z_0/a).
\label{eq:th1}
\end{equation}
}
\end{theorem}}

\revsS{
By Proposition~\ref{prop:9} (2) (d),
we can estimate each of the other terms appearing
in the power series expansion (\ref{eq-series})
by the truncated
Epstein-Hurwitz zeta function as follows.}

\begin{proposition}\label{prop:remainder}
\revsSSS{
For each $s \geq 3$, there exists a positive constant $M_s'$ such that 
\begin{eqnarray}
& & \sum_{(n_1, n_2)\in A_{\rho,N}}
\cE^{(s)}\left(
\frac{a}{(n_1 a - x_0)+(n_2 a - y_0)\ii}, \right. \nonumber \\
& & \qquad \qquad \qquad \qquad \qquad \qquad \left.
\frac{a}{(n_1 a - x_0)-(n_2 a - y_0)\ii}
\right) \nonumber \\
& & \qquad \le M_s' \zeta_{\rho, N}(s, -z_0/a).
\end{eqnarray}
}
\end{proposition}

\insOSSS{
\section{Remarks from Mathematical Viewpoints}\label{section6}
}

\insOSSS{
In this section, we give some remarks
on our results
from mathematical viewpoints.
}


\insS{
We \insSS{have} described multiple screw dislocations that are
parallel to each other
in the continuum picture in
Section~\ref{section2} as in
Definition~\ref{def:MDC}, using the section of a certain $S^1$-bundle
as defined in (\ref{eq:MDC}).
Although they were known as topological defects in \cite{HB,N},
in Proposition~\ref{prop:2.6}
we \insSS{have shown} that they are also expressed as a quotient space of the
path space or as an abelian covering of $\EE^2\setminus \cS$.
This means that our screw dislocations are
regarded as realizations of the abelian covering of $\EE^2\setminus \cS$
in the euclidean three space $\EE^3$.
}

\insS{
In Proposition~\ref{prop:SC-multi} of Section~\ref{section3},
the discrete picture of
such \insSS{multiple screw dislocations} in the SC lattice
\insSS{has been} obtained as the pullback of the fiber structure of
the multiple screw \insSS{dislocations} in \insSS{the} continuum picture.
It can naturally be extended to the case of the BCC lattice.
\insMMM{
However, in the case of the BCC lattice, the Burgers vector is parallel
to the $(1, 1, 1)$-direction up to automorphisms of the BCC lattice,
and its geometrical structure is a little bit complicated.
}
In Section~\ref{section4}, in order to treat
the BCC case, we expressed the
fiber structure with respect to the $(1,1,1)$-direction
using algebraic methods.
The geometrical properties are determined by purely algebraic
computations as in Lemma~\ref{lm:SC111}
and Proposition~\ref{prop:BCC}.
Apparently, such algebraic methods can be applied to
more general settings.
}

\insS{
In Section~\ref{section5}, we \insSS{have} computed the energy
of a screw dislocation. This is, in fact, related to
\revsSSS{the theory of harmonic maps} as follows.
The complete SC lattice is realized \revsSSS{in $\EE^3$ via}
the embeddings $\iota_{\AA_3,\delta}: \AA^a_3 (\cong\ZZ^3) \to \EE^3$
as in (\ref{eq:iota_A3}).
Such embeddings are parametrized by $\delta \in \EE^3$,
which can actually be considered to be
elements of the $3$-torus
group $T^3 = \RR^3/a\ZZ^3$.
This is because if $\delta - \delta' \in a\ZZ^3$
for $\delta, \delta' \in \EE^3$, then the images
of $\iota_{\AA_3,\delta}$ and $\iota_{\AA_3,\delta'}$
coincide.
Thus, a slight perturbation of the embedding $\iota_{\AA_3,\delta}$
gives rise to a map $\AA_3^a \to T^3$,
and its infinitesimal version gives a map
into the tangent space of $T^3$, which is identified with $\RR^3$.
Therefore, we can regard a realization
of the lattice $\AA_3^a$ as
a minimal point of a certain energy functional
related to a harmonic map whose target space is $T^3$,
and we see that such an energy functional is given by our
spring model by imitating the standard energy functional
as introduced in \cite{ES,U} from a discrete point of view.}


\insS{
On the other hand, a screw dislocation loses the symmetry except
for the third axis.
The \revsSSS{fibering} structure that we \insSS{have} discussed is related to the
action of the subgroup $S^1 = \{1\} \times
\{1\} \times S^1$ of $T^3$.
The configuration of a dislocation can also be regarded
as a minimal point of an energy functional.
Since the relevant map in the \revsSSS{theory of harmonic maps}
for the dislocations is from
$\fD_{\cS}^\SC$ to $T^3$ and $S^1$ acts on $\fD_{\cS}^\SC$,
in Section~\ref{section5}
we \insSS{have} computed the energy of a screw dislocation
by summing up the energy densities parametrized by \revsSSS{
$(n_1, n_2) \in \ZZ^2 \cong \fLSCp
= \pi_{\cS,\gamma}(\fD_{\cS}^\SC)$,}
where $\pi_{\cS,\gamma}$ is given in Definition~\ref{def:MDC}.
}

Then, such an energy is approximately obtained in terms of
the truncated Epstein-Hurwitz zeta function
(\ref{eq:tzeta}) in Theorem~\ref{thm:energy},
where the Epstein-Hurwitz zeta function
is defined by \cite{E,El,Tr} as \revsSSS{
\begin{equation}
\zeta(s, z_0)= \sum_{(n_1, n_2)\in \ZZ^2}
\frac{1}{((n_1+x_0)^2 +(n_2+y_0)^2)^{s/2}}
\label{eq:6.00}
\end{equation}}
for $z_0 = x_0 + y_0 \ii$. 
\revsS{
Note the fact that
the zeta function diverges at $s=2$, whereas
it converges at $s>2$, which
implies that 
the principal part of the elastic energy 
$E^{\mathrm{(p)}}_{\rho,N}(x_0,y_0)$
described in Theorem~\ref{thm:energy}
diverges for $N \to \infty$, whereas
each of the other terms in the power series expansion 
(\ref{eq-series}) of the energy
converges by Proposition~\ref{prop:remainder}.
Since the energy density \revsSSS{
$\cE_{n_1, n_2}$}
comes from the \revsSSSS{real} analytic function
$\cE(w, \overline{w})$ in Proposition~\ref{prop:9},
it is expected that
the elastic energy
$E_{\rho,N}(x_0,y_0)$ 
also diverges for $N \to \infty$, i.e.,
$$
E_{\rho,N}(x_0,y_0) \to \infty, \ \mbox{ for } N \to \infty,
$$
although we have been unable to prove this conjecture in
the present paper.
}

\insS{
To study the dependence of $\zeta(s, z_0)$ on $z_0$ is a very
important problem as the Hurwitz zeta function
$$
\zeta(s,q):=\displaystyle{
\sum_{n=0}^\infty (n + q)^{-s}}
$$
has an interesting dependence on $q$.
For example, the difference
\insOS{
$\zeta(s, -z_0/a)-\zeta(s, -z_0'/a)$} for $z_0, z_0' \in \CC$
with $z_0 \neq z_0'$
is related to the elastic energy of
our model. It should be noted, at least, that if $z_0' - z_0$ is
a lattice point of $a\ZZ^2$, then the difference must vanish.
}

\insS{
Since the Epstein-Hurwitz zeta function is based on the theory of
quadratic forms in \insSSS{euclidean} spaces and the
study of Minkowski \cite{Ca},
this fact might shed light on new mathematical
aspects of the lattice theory besides \cite{CS}.
Furthermore, since the SC lattice $\AA_2^a$ can be regarded as the
Gauss integers
$\ZZ[\sqrt{-1}]$, these results might also reveal an important
connection between the
number theory and the theory of dislocations.
}

\revsSSS{Finally, we have a remark on the appearance of the
Epstein-Hurwitz zeta functions. We used the zeta functions
to derive divergence and convergence results for
our energy described by certain power series. If we
just concentrate ourselves to such results, then
the usage of such zeta functions might not be necessary,
since divergence and convergence of the relevant
power series can be proved rather directly. However,
we are using the zeta functions, since they give a unified
treatment of the power series and make our
lines of discussions clearer in a rather essential way.
See also Lemma~\ref{lemmaA2} in Appendix.}

\bigskip

\insOSSS{
\section{Remarks and Discussions from Physical Viewpoints}\label{sec:PhysV}
}

\insOSSS{
In this section, we give some remarks
on our results
from physical viewpoints. They are basically physical
interpretations of our results, or some explicit
formulas. We also discuss
our results on energy of dislocations from various aspects.
This section is
mainly for readers with background in physics, and the
contents will be described
without mathematical rigorousness. We loosely use the logarithm
function as a multiple valued function.
}

\medskip

\insOSSS{
\subsection{Configuration description}\label{subsection621}
}

\insOSSS{In this subsection, we exhibit formulas
for a single screw dislocation of the BCC lattice, giving explicit
coordinates of the lattices after the dislocation. We also
discuss some relationships to known researches.}

\insS{
The results in Section~\ref{section2}
are basically well-known, e.g., in
\cite[Chap.~4]{HB}.
Let us consider multi-screw dislocations,
whose dislocation lines are all parallel to the
$x_3$-direction and are given by
points $z_i$ in $\cS_+$ for ``positive'' screw directions, and
by points $z_j$ in $\cS_-$ for ``negative'' ones, where
$\cS_+$ and $\cS_-$ are disjoint finite subsets of the complex plane $\CC$.
Set $\cS = \cS_+ \bigcup \cS_-$.
Setting the lattice unit $d$,
we \insSS{have seen}
that the $x_3$-coordinates of the points in the dislocations are given by
\begin{equation}
\frac{d}{2\pi\ii}
\log
\left(\gamma
\prod_{z_i\in \cS_+}\frac{z-z_i}{|z-z_i|}\cdot
\prod_{z_j\in \cS_-}\frac{\overline{z-z_j}}{|z-z_j|} \right)
 \mbox{ \rm for }
z \in \EE^2 \setminus \cS = \CC \setminus \cS
\label{eq:6.01}
\end{equation}
for some $\gamma \in S^1$, where
$\overline{z-z_j}$ is the complex conjugate of $z - z_j$.
\insSS{It consists of solutions to the Laplace equations
\begin{equation}
\frac{\partial}{\partial \bar z}
\frac{\partial}{\partial z}
\log \frac{z-z_k}{|z-z_k|} = 0 \ \mbox{ on }\
z \in \CC \setminus \cS,
\label{eq:6.02}
\end{equation}
for $z_k \in \cS$. }
In other words, the dislocations are obtained as a set of minimal
points of the elastic energy under a certain boundary condition
\cite{HB,N}.
}

\insS{
Based on the result (\ref{eq:6.01}) in Section~\ref{section2},
we \insSS{have described} the discrete picture of the dislocation in the SC lattice
in Subsection~\ref{subsec:sdSC};
at a point \revsSSS{$(n_1, n_2)$ of $\ZZ^2 \subset \CC$,} the
lattice points are given by
\revsSSS{
\begin{eqnarray*}
& & \Biggl(n_1 a, n_2 a, \, \frac{a}{2\pi\ii} \cdot \\
& & \log\biggl(\gamma \prod_{z_i \in \cS_+}
        \frac{(n_1 a + n_2 a \ii) - z_i}
        {|(n_1 a + n_2 a \ii) - z_i|}
        \prod_{z_j \in \cS_-}
          \frac{\overline{(n_1 a + n_2 a \ii) - z_j}}
         {|(n_1 a + n_2 a \ii) - z_j|} \biggr)\Biggr) , \\
& &    \hspace*{9cm} (n_1 a, n_2 a) \in \fLSCp,
\end{eqnarray*}
where we have set $\bar \delta = (0, 0)$.}
These are realized as points
in the configuration of the continuum
picture. In other words, these are also based on the solutions \insSS{to}
the Laplace equation (\ref{eq:6.02}).
}

\insS{
For the BCC case, the dislocation layers are split into three types.
Usually, the configuration has been discussed geometrically;
however, in this
article, we \insSS{have shown}
the fact by means of an algebraic method.
Experimentally, it is known that a dislocation line in a real
material \insSS{may not be a straight} line, but is
a curve in the $3$-dimensional euclidean space $\EE^3$ \cite{HB}.
Thus, we need to deal with such a curved dislocation line
mathematically.
Our algebraic approach might enable us to handle such a curve
in a lattice locally which could be a part of a curved dislocation
line.
We should emphasize that our method is novel
in this field of study.
}

\insS{
Using the result (\ref{eq:6.01}) in Section~\ref{section2},
we \insSS{have described} the discrete picture of the dislocation in the BCC lattice
in Subsection~\ref{subsection4.4}.
Since we \insSS{have shown}
it in Proposition~\ref{prop:multiBCC} directly,
let us rewrite it only for
a single screw dislocation here:
\revsSSS{
\begin{eqnarray*}
& & \mbox{The first layer: }\\
& & \left(\sqrt{2}\ell_1 a+\sqrt{2}\ell_2 a/2, \sqrt{6}\ell_2 a/2,
 \frac{\sqrt{3}a}{4\pi\ii}\log \frac{L_0}{|L_0|}\right) \\
& & \qquad \mbox{with } L_0 = \sqrt{2}\ell_1 a+\sqrt{2}\ell_2 a/2-x_0
    +(\sqrt{6}\ell_2 a/2 - y_0)\ii \\
& & \qquad \mbox{\rm for } x = \ell_1(a_1-a_3) + \ell_2(a_2-a_3) \in \fLBCCd^{(0)}, \\
& & \mbox{The second layer: } \\
& & \Biggl(\sqrt{2}\ell_1 a+ \sqrt{2}a/2 +\sqrt{2}\ell_2 a/2,
          (\sqrt{6}\ell_2 a-\sqrt{6}a/3)/2, \\
& & \qquad \qquad \qquad \qquad \qquad \qquad \qquad
\frac{\sqrt{3}a}{4\pi\ii}\log \frac{L_1}{|L_1|} - \frac{a}{2\sqrt{3}} \Biggr) \\
& & \qquad \mbox{with } L_1 = \sqrt{2}\ell_1 a+\sqrt{2}a/2+\sqrt{2}\ell_2 a/2-x_0 \\
& & \qquad \qquad \qquad \qquad \qquad \qquad \qquad
+((\sqrt{6}\ell_2 a-\sqrt{6}a/3)/2 - y_0)\ii \\
& & \qquad \mbox{\rm for } x = \ell_1(a_1-a_3) + \ell_2(a_2-a_3)
+ a_1-b \in \fLBCCd^{(1)}, \\
& & \mbox{The third layer: } \\
& & \Biggl(\sqrt{2}\ell_1 a+ \sqrt{2}a/2 +\sqrt{2}\ell_2 a/2,
          (\sqrt{6}\ell_2 a+2\sqrt{6}a/3)/2, \\
& & \qquad \qquad \qquad \qquad \qquad \qquad \qquad
\frac{\sqrt{3}a}{4\pi\ii}\log \frac{L_2}{|L_2|} - \frac{a}{\sqrt{3}} \Biggr) \\
& & \qquad \mbox{with } L_2 = \sqrt{2}\ell_1 a+ \sqrt{2}a/2+\sqrt{2}\ell_2 a/2-x_0 \\
& & \qquad \qquad \qquad \qquad \qquad \qquad \qquad
+((\sqrt{6}\ell_2 a+2\sqrt{6}a/3)/2 - y_0)\ii  \\
& &  \qquad \mbox{\rm for } x = \ell_1(a_1-a_3) + \ell_2(a_2-a_3) + a_1+a_2-b
\in \fLBCCd^{(2)},
\end{eqnarray*}
where we have set $\bar \delta = (0, 0)$.}}

\insS{Recently the configurations of the dislocations are studied in terms
of the ab-initio computations \cite{Cl}; however,
as mentioned above, the boundary
condition is crucial in the study of dislocations.
Since our description of the configuration is for the region
far from the dislocation line, the configuration
should obey the classical mechanics as the continuum theory of
dislocation.
Even for the ab-initio computations of the core structure of a
dislocation, our results may provide
data for their boundary conditions.
Furthermore, recently crystal structures can be observed directly
and are analyzed in terms of the number theory \cite{ISCKI} as well.
In a similar sense, our results might provide new
viewpoints for the study of dislocations.
}

\insS{
Furthermore,
even if the thermal fluctuation is \insSSS{locally larger than the elastic energy},
our study shows that
the topological defect \insSS{cannot} be neglected, since the contour integral
of the configurations of atoms along a circle around the dislocation
line \insSSS{gives} the topological invariance, which is described in
Remark~\ref{rmk:2.3}
and Proposition~\ref{prop:2.6} mathematically.
\insSSS{For real materials,}
we must consider other effects, e.g., various dislocations,
bend of dislocation lines, etc.; however,
some of the properties based on topological arguments
\insSSS{must} be preserved in the discrete description as mentioned in
Sections~\ref{section3} and \ref{section4}.
We note that the relation to the path space in Section~\ref{section2}
is robust, since the abelian covering is associated with
the contours of certain integrals as in the case of
Riemann surfaces \cite{KMP}.
}

\medskip

\ins{
\subsection{Energy description}
}

\ins{
In Section~\ref{section5}, we \insSS{have}
obtained the strain or elastic energy of the
screw dislocation.
}
\insOSSS{
In this subsection, we give five remarks, I--V, on the energy
of screw dislocations.
In I, we give a remark on the choice
of the edges for springs in our model.
In II, we discuss the relationship between
the energy computation in the continuum picture and that in the
discrete one, and we also give physical interpretations of our
result. III is about
the core region, IV is about the BCC lattice case, and in V we
apply our result to show certain finiteness of the energy for a pair
of parallel screw dislocations with opposite directions.
}

\medskip

\insS{
I. \insOSSS{\textbf{Choice of edges.}}
In the computation, we \insSS{have}
considered the spring model of the
SC lattice by assuming that we have springs for a certain set
of edges of the lattice graph.
As in Lemma~\ref{lm:elasap}, the increase in length caused by the
dislocation strongly depends on the direction:
as (\ref{eq:approx_of_delta}) shows,
it is of order one with respect to
the height difference $\epsilon$
for the edges including the \insSS{$a_3$-direction}, while
it is of order two for the other edges.
Note that the latter can be neglected in the
approximate computation of the relevant energy.
}

\revsSSS{
This means that even if we add, for example, a spring for each edge
$$[(n_1, n_2,n_3), (n_1+1, n_2+1, n_3+1)]$$
etc.,
the \revsSSS{approximate} energy basically remains the same as that given in
Theorem~\ref{thm:energy}.
}

\medskip

\insS{
II. \insOSSS{\textbf{Continuum picture versus discrete one.}}
For simplicity, let us suppose that the
position of the dislocation line corresponds to
$(x_0, y_0) = (0, 0)$.
As we \insSS{have} investigated the elastic energy for the annulus region
\insSSS{$$
R_{\rho, N} :=
\left\{(x,y)\in \EE^2 \,|\,
\rho a < \sqrt{x^2+y^2} < N a\right\}
$$}of the dislocation for the discrete picture
in Theorem~\ref{thm:energy}, let us also consider
its counterpart for the continuum picture.
It is well-known that the strain energy
$E_{\rho,N}^c$ of the screw dislocation
in the annulus region \insM{per unit length along $(0,0,1)$-direction},
in continuum picture,
is expressed by
\insM{
$$
E_{\rho,N}^c = \frac{a^2 G}{4\pi} \log N / \rho,
$$
where $G$ is the shear modulus
and $a$ appears as the length
of the Burgers vector (see \cite[eq.(4.20)]{HB}).}
This is given by
\insM{
$$
\int_{R_{\rho,N}}\frac{1}{2}\frac{a^2G}{(2\pi)^2} \frac{1}{x^2+y^2} \, dx dy
=\frac{a^2G}{8\pi^2}\int_{\rho a}^{Na} \frac{1}{r^2} r\,dr \int^{2\pi}_0 d\theta
=\frac{a^2G}{4\pi} \log N/\rho,
$$
}
where the integrand comes from the elastic energy
\begin{eqnarray*}
& & \left(\frac{1}{2\pi\ii}\frac{\partial}{\partial x}
\log \frac{x+y \ii}{|x+y \ii|} \right)^2
+\left(\frac{1}{2\pi\ii}\frac{\partial}{\partial y}
\log \frac{x+y \ii}{|x+y \ii|} \right)^2 \\
& = & -\frac{4}{(2\pi)^2}
\left(\frac{\partial}{\partial z}\frac{1}{2}
\log \frac{z}{\bar z} \right)
\left(\frac{\partial}{\partial \bar z}\frac{1}{2}
\log \frac{z}{\bar z} \right) \\
& = &
\frac{1}{(2\pi)^2}
\frac{1}{z\bar z}=
\frac{1}{(2\pi)^2}
\frac{1}{x^2+y^2}
\end{eqnarray*}
for $z= x +y \ii$ (see \cite{HB,LL,N}).
}

\insM{Note that
the modulus $G$ is directly \insOSSS{related}
to the spring constant $k_d/a$ in
our discrete model,
since} \insSSS{the integral for computing the strain
energy $E_{\rho,N}^c $ in the continuum picture
corresponds to the summation over lattice points
in the region in the discrete picture, and each term is consistent with
each other as shown by equation} \revsS{(\ref{eq15}).}
\insOS{More precisely,
we have considered only a layer of length $a$
to evaluate the energy $E_{\rho,N}(x_0,y_0)$
in our discrete model, whereas $E_{\rho,N}^c$
is for the unit length along the dislocation line.
\insOSS{
By putting $G=k_d/a$,
$E_{\rho,N}(0,0)/a$ corresponds to $E_{\rho,N}^c$;
in fact, by dividing the energy
by the square $a^2$ of the length of the Burgers vector,
we have that
$$\lim_{a \to 0} (1/a)E_{\rho,N}(0,0)/a^2 =
\lim_{a \to 0} E_{\rho,N}^c/a^2$$}holds.
This means that our spring model, in discrete picture, is
consistent with the known dislocation theory in continuum picture,
so that our model is plausible in such a sense.}
\insOSS{
However, as mentioned in I, the correspondence between
the spring constant $k_d$ and the shear modulus $G$ does
depend on the choice of the edges for
springs. If we employ additional edges with
their own spring constants, then
the resulting elastic energy may have a different constant,
and thus the correspondence might be modified.
}

\insS{
For $N \to \infty$,
the above energy $E_{\rho,N}^c$ in continuum picture diverges.
This should be compared with the discrete picture:
\revsSSS{the principal part of the elastic energy
$E^{(\mathrm{p})}_{\rho,N}(x_0,y_0)$} in Theorem~\ref{thm:energy}
also diverges \insSS{for $N \to \infty$}
due to the property of the Epstein-Hurwitz
zeta function \cite{Tr} \revsSSS{(see also Section~\ref{section6})}.
}

\medskip

\insS{
III. \insOSSS{\textbf{Core region.}}
As mentioned \insSSS{just before} the definition of $A_\rho$ in (\ref{eq:Arho}),
our description shows that
there is a criterion for the core region of a screw dislocation
from a viewpoint of elastic energy.
}

\medskip

\insS{
IV. \insOSSS{\textbf{BCC lattice case.}}
The investigation in Section~\ref{section5}
can also be applied to
the case of the BCC lattice with the help of
Proposition~\ref{prop:singleBCC},
although it might be complicated.
}

\medskip

\insS{
V. \insOSSS{\textbf{Double screw dislocation case.}}
Finally, let us demonstrate an application of
our model using the \insSS{zeta} function as follows.
The total strain energy of \insOSSS{a pair of screw} dislocations
whose dislocation
lines are parallel to each other and correspond
to $\cS_+=\{(x_0, y_0)\}$
and $\cS_-=\{(x_0, -y_0)\}$ with $y_0 \neq 0$,
is well-known in the continuum picture
as follows:
\begin{eqnarray}
& & E_{\rho,N}^c(\cS) \nonumber \\
& = &
C \int_{R_{\rho,N}'} \frac{1}{(x-x_0)^2 + (y-y_0)^2} dx dy
+C \int_{R_{\rho,N}'} \frac{1}{(x-x_0)^2 + (y+y_0)^2} dx dy
\nonumber \\
& & \quad
-2C \int_{R_{\rho,N}'}
 \frac{(x-x_0)^2 + (y+y_0)(y-y_0)}
 {((x-x_0)^2 + (y-y_0)^2)((x-x_0)^2 + (y+y_0)^2)} dx dy,
\label{eq:EcS}
\end{eqnarray}
where $0 < \rho < N$, $\cS = \cS_+ \cup \cS_-$,
$R_{\rho, N}'
:= R_{\rho, N, (x_0,y_0)} \cap R_{\rho, N, (x_0,-y_0)}$,
and
$$
R_{\rho, N, (x_0, \pm y_0)} :=
\left\{(x,y)\in \EE^2 \,|\, \rho a < \sqrt{(x - x_0)^2 +
(y \mp y_0)^2} < N a\right\}
$$
(see \cite{HB,N}).
Then we can show that
for $\rho a > 2 |y_0|$,
$E_{\rho,N}^c(\cS)$ is finite and it \insSSS{converges}
for $N \to \infty$ (see \insSSS{appendix}).
}

It is important to show that this phenomenon occurs also
for the discrete picture,
or for the SC lattice.
Let us evaluate the elastic energy in our model
used in Section~\ref{section5}
for the double screw dislocation \revsSSS{
\begin{eqnarray*}
\check\sigma_{z_0,\bar z_0}(n_1 a,n_2 a)
& = & \check\sigma_{\{z_0,\bar z_0\}, 1}(n_1 a,n_2 a)
\\
& = & \frac{n_1 a - x_0+ (n_2 a - y_0) \ii}
{\abs{n_1 a - x_0+ (n_2 a - y_0) \ii}}
 \frac{n_1 a - x_0- (n_2 a + y_0) \ii}
{\abs{n_1 a - x_0- (n_2 a + y_0)\ii}},
\end{eqnarray*}
instead of
$\check\sigma_{z_0}(n_1 a,n_2 a)$,
where $z_0 =x_0 + y_0 \ii$, $x_0, y_0 \in \RR$.}

\revsMM{
Let us put
$B_{\rho,N}'
:= B_{\rho, N, (x_0,y_0)} \cap B_{\rho, N, (x_0,-y_0)}$
for
\revsMMM{
$$
B_{\rho, N, (x_0, \pm y_0)} :=
\left\{(n_1,n_2)\in \ZZ^2 \,|\, \rho a <
\sqrt{(n_1 a -x_0)^2+(n_2 a  \mp y_0)^2} < N a \right\}.
$$
}
Then the elastic energy of the configuration for our $\cS$
is given by} \revsSSS{
$$
E_{\rho,N}(\cS):= \sum_{(n_1, n_2)\in B_{\rho,N}'}
\cF_{n_1,n_2},
$$
where $\cF_{n_1,n_2}$ is the energy density defined by
\begin{eqnarray}
\cF_{n_1,n_2}& := & 
\frac{1}{2}k_p
\biggl(\left(\Delta_{n_1, n_2}^{(1)}\right)^2
+\left(\Delta_{n_1, n_2}^{(2)}\right)^2 \biggr)
\nonumber \\
& & \quad + \frac{1}{2}k_d
\biggl(
\left(\Delta_{n_1, n_2}^{d(1, +)}\right)^2
+\left(\Delta_{n_1, n_2}^{d(2, +)}\right)^2
+\left(\Delta_{n_1, n_2}^{d(1, -)}\right)^2  \nonumber \\
& & \quad
+\left(\Delta_{n_1, n_2}^{d(2, -)}\right)^2
+\left(\Delta_{n_1, n_2}^{d(+)}\right)^2
+\left(\Delta_{n_1, n_2}^{d(-)}\right)^2\biggr).
\raisebox{0mm}[4mm][4mm]{}
\nonumber
\end{eqnarray}
Here, $\Delta_{n_1, n_2}^{(1)}$ etc.\ denotes
the difference in length between nearby lattice points
before/after the double dislocation
in the designated direction at the lattice point
corresponding to $(n_1, n_2)$, as in
Subsection~\ref{subsection5.1}.
As in Proposition~\ref{prop:9}, we see easily that
there exists a \revsSSSS{real} analytic function
$\cF(w_+, \overline w_+, w_-, \overline w_-)$
of \revsSSSS{$w_\pm$ and $\overline{w}_\pm \in \CC$} such that
$\cF_{n_1,n_2}$ is given by
the substitution
$$
w_\pm =\frac{a}{(n_1 a - x_0) + (n_2 a \pm y_0) \ii}. 
$$
It is also natural to consider its power series expansion in $w_\pm$ and 
$\overline{w}_\pm$ as in Proposition~\ref{prop:9}.
However, it is expected that
its leading term plays an essential role
as in Theorem~\ref{thm:energy}. Thus, in the following,
let us evaluate the principal part of the energy.
As we have seen in Subsection~\ref{subsection5.1},
we may concentrate ourselves to
$\Delta_{n_1, n_2}^{d(i, \pm)}$,
which essentially contribute
to the energy. 
}

\revsSSS{
In the following, 
$\epsilon_{n_1,n_2}^{(i)}$
will denote the value corresponding to the 
vertical elongation as in Subsection~\ref{subsection5.1}.
By straightforward calculations, we have
\begin{eqnarray*}
\Delta_{n_1, n_2}^{d(i, \pm)}
& = & \pm \frac{1}{\sqrt{2}} \epsilon_{n_1,n_2}^{(i)}
+ o\left(
\frac{a}{\sqrt{(n_1 a -x_0)^2+(|n_2| a -|y_0|)^2}}\right), \\
\epsilon_{n_1,n_2}^{(1)}
& = & \frac{a}{2\pi}
    \left(
-\frac{ a(n_2 a -y_0)}
          {(n_1 a -x_0)^2 + (n_2 a - y_0)^2}
+\frac{ a(n_2 a +y_0)}
          {(n_1 a -x_0)^2 + (n_2 a + y_0)^2}\right)
\\
& & \quad + o\left(
\frac{a}{\sqrt{(n_1 a -x_0)^2+(|n_2| a - |y_0|)^2}}\right), \\
\epsilon_{n_1,n_2}^{(2)}
& = & \frac{a}{2\pi}
    \left(
-\frac{a(n_1 a -x_0)}
          {(n_1 a -x_0)^2 + (n_2 a - y_0)^2}
+\frac{a(n_1 a -x_0)}
          {(n_1 a -x_0)^2 + (n_2 a + y_0)^2} \right)
\\
& & \quad + o\left(
\frac{a}{\sqrt{(n_1 a -x_0)^2+(|n_2| a - |y_0|)^2}}\right). \\
\end{eqnarray*}}
\revsSS{Then, we have}
\revsSSS{
\begin{eqnarray*}
\cF_{n_1,n_2}
& = & k_d
\left[
 \frac{1}{2} (\epsilon_{n_1,n_2}^{(1)})^2
 +\frac{1}{2} (\epsilon_{n_1,n_2}^{(2)})^2\right. \\
& & \left.\quad + o \left(\frac{a^2}{(n_1 a -x_0)^2+(|n_2| a - |y_0|)^2}\right)
\right]\\
& = &
\frac{1}{8\pi^2}k_d
\left[
\frac{a^4}{(n_1 a -x_0)^2+(n_2 a -y_0)^2} \right. \\
& & \quad +\frac{a^4}{(n_1 a -x_0)^2 +(n_2 a +y_0)^2}\\
&& \quad -2\frac{a^2
((n_1 a -x_0)^2 +(n_2 a +y_0)(n_2 a -y_0)) }
{((n_1 a -x_0)^2 +(n_2 a -y_0)^2)
((n_1 a -x_0)^2 +(n_2 a +y_0)^2)} \\
& & \quad
    + \left.o\left(
\frac{a^2}{(n_1 a -x_0)^2+(|n_2| a - |y_0|)^2}\right)
\right].
\end{eqnarray*}
}
Let us assume that $\rho > 2|y_0|$. Then,
using the property of the Epstein zeta function \cite[Cor.~1.4.4]{Tr},
we can show that
\revsSSS{the principal part of the energy $E_{\rho,N}(\cS)$,
\begin{eqnarray*}
E^{(\mathrm{p})}_{\rho,N}(\cS) 
& := &
\frac{1}{8\pi^2}k_d
 \sum_{(n_1, n_2) \in B_{\rho,N}'}\left[
\frac{a^4}{(n_1 a -x_0)^2+(n_2 a -y_0)^2} \right. \\
& & \quad +\frac{a^4}{(n_1 a -x_0)^2 +(n_2 a +y_0)^2}\\
& & \quad \left. -2\frac{a^2
((n_1 a -x_0)^2 +(n_2 a +y_0)(n_2 a -y_0)) }
{((n_1 a -x_0)^2 +(n_2 a -y_0)^2)
((n_1 a -x_0)^2 +(n_2 a +y_0)^2)} \right],
\end{eqnarray*}}
is finite even for
$N \to \infty$ (see \insSSS{appendix}).
\revsSS{Thus, it is expected that
$E_{\rho,N}(\cS)$ is also finite 
even for $N \to \infty$ due to the reason similar to 
Proposition~\ref{prop:remainder},
although we have not been able to prove this conjecture so far, 
either.
}

\revsSS{
These results show that the properties
of zeta functions may enable us to evaluate the
discrete system in a rigorous manner.
We believe that our investigation is necessary
for clarifying discrete systems, e.g., in the framework of the classical
statistical mechanics.
}

\bigskip


\section*{Appendix}

\revsSS{
In this appendix, we show that the total energy
for the double screw dislocation discussed in Section~\ref{section6}
\insSSS{converges} for $N \to \infty$ 
in the continuum picture. We also show the counterpart
in the discrete picture
for the principal part.}

\insS{
\begin{lemma}
Suppose $\rho a > 2|y_0|$.
Then, the strain energy
$E_{\rho,N}^c(\cS)$ given in \textup{(\ref{eq:EcS})}
for the region $R_{\rho, N}'$ in the continuum picture
\insSSS{converges} for $N \to \infty$.
\end{lemma}
}

\begin{proof}
\insS{
Note that the integrand is given as
\insSSS{\begin{eqnarray*}
& & \frac{1}{(x - x_0)^2 + (y - y_0)^2} + \frac{1}{(x - x_0)^2 + (y + y_0)^2} \\
& & \quad - \frac{2 ((x - x_0)^2 + (y + y_0) \,
(y - y_0))}{((x - x_0)^2 + (y - y_0)^2 ) \,
((x - x_0)^2 + (y + y_0)^2 )} \\
&=& \frac{4 y_0 ^2}{ ((x - x_0)^2 + (y - y_0)^2 ) \,
( (x - x_0)^2 + (y + y_0)^2 )}
\end{eqnarray*}}and that $R_{\rho,N}' \subset R_{\rho', N, (x_0, 0)}$ for
some $\rho' > 0$ with $\rho' a > |y_0|$. Then, we have
\insSSS{
\begin{eqnarray*}
& & C^{-1} E_{\rho,N}^c(\cS) \\
& \leq & \int_{R_{\rho', N, (x_0, 0)}}
\frac{4 y_0 ^2}{ ((x - x_0)^2 + (y - y_0)^2 ) \,
( (x - x_0)^2 + (y + y_0)^2 )} dxdy \\
& =: & E_{\cS, 0}.
\end{eqnarray*}}We may assume that $y_0$ is positive.
By using the polar coordinate $(r, \theta)$ centered
at $(x_0, 0)$ such that
$x = x_0 + r \cos{\theta}$ and $y = r \sin{\theta}$,
we have
\begin{eqnarray*}
E_{\cS, 0} &=& \int_0 ^{2 \pi} \int_{\rho' a} ^{N a}
\frac{4 y_0 ^2 r}{ (r^2 - 2 y_0 r \sin{\theta} + y_0 ^2) \,
(r^2 + 2 y_0 r \sin{\theta} + y_0 ^2)} dr d\theta \\
&\leq& \int_0 ^{2 \pi} \int_{\rho' a} ^{N a}
\frac{4 y_0 ^2 r}{ (r - y_0) ^4 } dr d\theta \\
&=& 8 \pi y_0 ^2 \int_{\rho' a} ^{Na}
\left( \frac{1}{(r - y_0)^3} + \frac{y_0}{(r - y_0)^4} \right) dr \\
&=& 8 \pi y_0 ^2 \left( \frac{1}{2(\rho' a - y_0)^2}
+ \frac{y_0}{3(\rho' a - y_0)^3}
- \frac{1}{2(N a - y_0)^2} - \frac{y_0}{3(N a - y_0)^3} \right) \\
&\to& 8 \pi y_0 ^2 \left( \frac{1}{2(\rho' a - y_0)^2}
+ \frac{y_0}{3(\rho' a - y_0)^3} \right)
\end{eqnarray*}
as $N \to \infty$, which means that
$E_{\cS, 0}$ does not diverge for $N \to \infty$.
This completes the proof.}
\end{proof}

\insS{\begin{lemma}\label{lemmaA2}
Suppose \insSSS{$\rho a > 2|y_0|$}.
Then, the \revsMM{the principal part of the
elastic energy $E_{\rho,N}^{(\mathrm{p})}(\cS)$}
for the region $B_{\rho, N}$ in the discrete picture
\insSSS{converges} for $N \to \infty$.
\end{lemma}}

\begin{proof}
\revsSSS{
Set
\begin{eqnarray*}
I&:=&\frac{a^2}{(n_1 a -x_0)^2 +(n_2 a -y_0)^2}
+\frac{a^2}{(n_1 a -x_0)^2 +(n_2 a +y_0)^2} \\
&& \quad -2\frac{a^2
((n_1 a -x_0)^2 +(n_2 a +y_0)(n_2 a -y_0))}
{((n_1 a -x_0)^2 +(n_2 a -y_0)^2)
((n_1 a -x_0)^2 +(n_2 a +y_0)^2)} \\
& = & 4\frac{a^2 y_0^2 }
{((n_1 a -x_0)^2 +(n_2 a -y_0)^2)
((n_1 a -x_0)^2 +(n_2 a +y_0)^2)}.
\end{eqnarray*}
We may assume that $y_0 > 0$. Then,
for $n_2 \ge 0$, we have
$(n_2 a +y_0)^2 \ge (n_2 a -y_0)^2$
and
$$
I\le \frac{4a^2 y_0^2 }
{((n_1 a -x_0)^2 +(n_2 a -y_0)^2)^2},
$$
whereas for $n_2 \le 0$, we have
$(n_2 a +y_0)^2 \le (n_2 a -y_0)^2$ and
$$
I\le \frac{4a^2 y_0^2 }
{((n_1 a -x_0)^2 +(n_2 a +y_0)^2)^2}.
$$}
Furthermore, since for the Epstein zeta function, the value
\revsSSS{$$
\sum_{(n_1,n_2,\ldots,n_n) \in \ZZ^n\setminus\{(0,0,\ldots,0)\}}
\frac{1}{(n_1^2+n_2^2+\cdots+n_n^2)^{s/2}}
$$}is finite for $s>n$ \cite[Cor.~1.4.4]{Tr}, we see that
\revsSS{the principal part of the
elastic energy 
$E_{\rho,N}^{(\mathrm{p})}(\cS)$
converges for $N\to \infty$.}
\insSS{This completes the proof.}
\end{proof}

\insSS{Note that in Lemma~\ref{lemmaA2}, we need the
assumption $\rho > 2|y_0|$ in order to avoid the core
regions around the dislocation centers.}

\bigskip
\bigskip

{\bf{Acknowledgements: }}
{\small{
The authors would like to express their sincere gratitude
to all those who participated in the problem session
``Mathematical description of disordered structures in crystal''
in the Study Group Workshop 2015 held in Kyushu University
and in the University of Tokyo during July 29--August 4, 2015 \cite{O}.
They are grateful to
\insSSS{Professors}
\insOSS{Motoko Kotani, Takayuki Oda and Tatsuya Tate}
for helpful comments
\insOSS{and to Professor Shun-ichi Amari for sending them his works
\cite{A1,A2}.}
\revsSSS{Special thanks go to Professor Hiroyuki Ochiai
for his numerous essential comments which drastically improved the paper.}
\ins{The 2nd author has been supported in part by
JSPS KAKENHI Grant Number \insOSSS{JP16K05187}.
}
The 4th author has been supported in part by
\insOSSS{JSPS KAKENHI Grant Numbers JP15K13438, JP17H06128}.
The second author thanks Professor Kenichi Tamano
for critical discussions for \insS{an earlier version} of this article.
\insS{
The authors also thank Professor Yohei Kashima and Professor Akihiro
Nakatani for pointing out the analysis of a pair of dislocations.
Thanks are also to
the anonymous referees for critical suggestions
and comments for \insOSSS{previous versions}, especially
for the section for physicists, Remarks~\ref{rmk:Snake1}
and \ref{rmk:Snake2}, and the expression
(\ref{eq:th1}) of the energy for the annulus region.
}
}}

\bigskip
\bigskip
\bigskip



\bigskip

\bigskip

\noindent
H.~Hamada:\\
National Institute of Technology, Sasebo College,\\
1-1 OkiShin-machi, Sasebo, Nagasaki, 857-1193, JAPAN\\
\\
S.~Matsutani:\\
National Institute of Technology, Sasebo College,\\
1-1 OkiShin-machi, Sasebo, Nagasaki, 857-1193, JAPAN\\
\\
J.~Nakagawa:\\
Mathematical Science \& Technology Research Labs,\\
Nippon Steel \& Sumitomo Metal Corporation,\\
20-1 Shintomi, Futtsu, Chiba, 293-8511, JAPAN\\
\\
O.~Saeki:\\
Institute of Mathematics for Industry,\\
Kyushu University,\\
744 Motooka, Nishi-ku, Fukuoka 819-0395, JAPAN\\
\\
M.~Uesaka:\\
Graduate School of Mathematical Sciences,\\
The University of Tokyo,\\
3-8-1 Komaba, Meguro-ku, Tokyo, 153-8914, JAPAN\\
moved to\\
Research Institute for Electronic Science,\\
 Hokkaido University,\\
N12W7, Kita-Ward, Sapporo, 060-0812, Japan.\\
\\

\end{document}